\pdfoutput=1
\documentclass[times]{speauth}

\usepackage[usenames,dvipsnames]{xcolor}

\usepackage{etoolbox}
\usepackage{pbox}

\makeatletter
\def\ps@pprintTitle{%
  \let\@oddhead\@empty
  \let\@evenhead\@empty
  \let\@oddfoot\@empty
  \let\@evenfoot\@oddfoot
}
\makeatother
\usepackage{numcompress}
\usepackage{setspace}
\usepackage[colorlinks,bookmarksopen,bookmarksnumbered,citecolor=red,urlcolor=red]{hyperref}
\usepackage[T1]{fontenc}
\usepackage[utf8]{inputenc}
\usepackage{algorithm}
\usepackage{paralist}
\usepackage{multicol,multirow}
\usepackage{arydshln}
\usepackage{colortbl}

\usepackage{listings}

\usepackage{algorithmic}
\usepackage{subfig,subfloat}

\usepackage{amssymb}
\usepackage{booktabs}
\usepackage{amsmath}
\usepackage{mathtools}
\usepackage{verbatim}

\newcommand\given[1][]{\:#1\lvert\:}
\usepackage{threeparttable}

\usepackage{subfloat,subfig,graphicx,url}
\graphicspath{{./figures/}{../gnuplot/}}
\usepackage{pgf}
\usepackage{tikz}
\usetikzlibrary{arrows}

\newcommand{\treename}{\textsc{PM+}}
\newcommand{\nameofthescheme}{\textsc{PM+-\allowbreak{}Multilinear}}
\newcommand{\fullnameofthescheme}{\textsc{Pseudo+\allowbreak{}Mersenne-\allowbreak{}Multilinear}}
\newcommand{\treenamesingle}{\textsc{PM+32}}
\newcommand{\nameoftheschemesingle}{\textsc{PM+-\allowbreak{}Multilinear32}}

\newcommand{\treenamedbl}{\textsc{PM+64}}
\newcommand{\nameoftheschemedbl}{\textsc{PM+-\allowbreak{}Multilinear64}}

\newcommand\xor{\veebar}

\usepackage{todonotes}

\usepackage{flushend}

\usepackage{amsthm}

\newtheorem{definition}{Definition}
\newtheorem{definition*}{Definition}
\newtheorem{remark*}{Remark}
\newtheorem{lemma}{Lemma}

\begin{document}

\runningheads{D. Ivanchykhin, S. Ignatchenko, D. Lemire}{Regular and almost universal hashing}

%\shortauthors{D. Ivanchykhin, S. Ignatchenko  and D. Lemire}

\title{Regular and almost universal hashing: an efficient implementation}

\author{D. Ivanchykhin\affil{1}, S. Ignatchenko\affil{1}, D. Lemire\affil{2}}

\address{\affilnum{1}OLogN Technologies AG, Triesen, Liechtenstein \break
\affilnum{2}LICEF Research Center, TELUQ, Montreal, QC, Canada}

%\author{Dmytro Ivanchykhin} \email{di@o-log-n.com} \affiliation{OLogN Technologies AG,
%Landstrasse 123,
%LI-9495 Triesen,
%Liechtenstein}

\cgsn{Natural Sciences and Engineering Research Council of Canada}{261437}
\corraddr{Daniel Lemire, LICEF Research Center, TELUQ,
Universit\'e du Qu\'ebec,
5800 Saint-Denis,
Office 1105,
Montreal (Quebec),
H2S 3L5 Canada. Email: lemire@gmail.com}

%\shortauthors{D. Ivanchykhin, S. Ignatchenko  and D. Lemire}

%\author{Dmytro Ivanchykhin} \email{di@o-log-n.com} \affiliation{OLogN Technologies AG,
%Landstrasse 123,
%LI-9495 Triesen,
%Liechtenstein}
%
%
%\author{Sergey Ignatchenko} \email{sergey@ignatchenko.com} \affiliation{OLogN Technologies AG,
%Landstrasse 123,
%LI-9495 Triesen,
%Liechtenstein}
%
%\author{Daniel Lemire}
%\affiliation{LICEF, Universit\'e du Qu\'ebec,
% 5800 Saint-Denis, Montreal, QC, H2S 3L5 Canada}
%
%\keywords{Universal hashing; Regularity; Denial-of-Service attacks; Finite field computations }

\begin{abstract}
Random hashing can provide  guarantees regarding the performance of data
structures such as hash tables---even in an adversarial setting.
Many existing families of hash functions are universal:
given two data objects, the probability that they have the same
hash value is low given that we pick hash functions at random.
However, universality fails to ensure
that all hash functions are well behaved.
We might further require regularity: when picking
data objects at random they should have a low probability
of having the same hash value, for any fixed hash function.
We present the efficient implementation of a family of non-cryptographic hash functions
 (\treename{}) offering good running times, good memory usage as well as distinguishing theoretical guarantees:
almost universality
and component-wise regularity. On a variety
 of platforms, our implementations are
 comparable to the state of the art in
 performance. On recent Intel processors,
  \treename{} achieves a speed of
   4.7~bytes per cycle for 32-bit outputs and
3.3~bytes per cycle for 64-bit outputs. We
 review vectorization through SIMD
 instructions (e.g., AVX2)  and
 optimizations for superscalar execution.
\end{abstract}

\maketitle

\lstdefinestyle{customc}{%
  belowcaptionskip=1\baselineskip,
  breaklines=true,
  xleftmargin=\parindent,
  language=C,
  showstringspaces=false,
  basicstyle=\small\ttfamily,
  keywordstyle=\bfseries\color{green!40!black},
  commentstyle=\itshape\color{purple!40!black},
  identifierstyle=\bfseries\color{black},
  stringstyle=\color{orange},
   morekeywords={uint64_t,uint32_t,__m256i,__m128i},
}

\lstset{escapechar=@,style=customc}

\section{Introduction}

Hashing is ubiquitous in software. For example, most programming languages support hash tables, either directly, or via libraries. However,  while many computer science textbooks consider random hashing, most software libraries  use deterministic (i.e., non-random) hashing.

 A hash function maps data objects, such as strings,  to fixed-length values (e.g., 64-bit integers). We often expect data objects to be mapped
evenly over the possible hash values. Moreover,
we expect collisions to be unlikely:  there is a \emph{collision} when two
objects are mapped to the same hash value.
Hash tables can only be expected to offer
constant time query performance when collisions
are infrequent.

When hashing is deterministic, we pick one hash function  once and for all. It is even customary for this hash function to be common knowledge.
Moreover, the objects
being hashed are typically not random, they could even be provided by an adversary.
Hence, an adversary can cause
many collisions that could translate into a denial-of-service (DoS) attack~\cite{Crosby:2003:DSV:1251353.1251356,ocert2011003,Oorschot2006}.

In random hashing,  we regularly pick a new hash function at random from a family of hash functions.  With such random hashing, we can  bound the collision probability between two objects, even if the objects are chosen by an adversary. By using random hashing, programmers might produce more secure software and avoid DoS attacks. Maybe for this reason, major languages have adopted random hashing. Python (as of version~3.3), Ruby (as of version~1.9) and Perl (as of version~5.18) use random hashing by default~\cite{Orton2013}.
Unfortunately, %to our knowledge,
these languages fail to offer a theoretical guarantee regarding collision probabilities.

A family of hash functions having a low collision probability given two objects chosen by an adversary may contain terrible hash functions: e.g., hash functions mapping all objects to the same value (see \S~\ref{sec:regularity}). In practice, such bad hash functions can be reused over long periods of time: e.g., a Java program expects the hash value of a given object to remain the same for the duration of the program. An adversary could detect that a bad hash function has been selected and launch a successful attack~\cite{Handschuh2008,Saarinen2012}. Thus we should ensure that all hash functions in a family can be safely used.
We believe that regularity might help in this regard: a function is regular if all hash values are equally likely given that we pick the data objects at random. Indeed, regularity implies that the probability that two objects picked at random have the same hash value is low. We generalize regularity to component-wise regularity by considering same-length strings that differ by one character. We want to minimize the probability that any two such strings have the same hash value.

In the following sections, we describe a practical approach toward generating non-cryptographic hash functions for arbitrary
objects that have low probabilities of collision when picking hash values at random (i.e., universality) and
low probabilities when picking data objects at random (i.e., regularity).

It is not
difficult to construct such families: e.g., the family $h(x) = a x \bmod \,p$ for $p$ prime and an integer $a$ picked randomly in $[1,p)$ is regular and almost universal over integers in $[0,p)$. However, our objective is to implement a practical solution in software that provides competitive speed. In particular, we wish to hash arbitrarily  long strings of bytes, not just integers in $[0,p)$, and in doing so,
we wish to make the best possible use of current processors. To achieve our goals, we  use affine functions over a suitable finite field to hash blocks of machine words. We choose the finite field so as to make the operations  efficient. We then use a tree construction to hash long strings. The resulting family of hash functions is called \treename{}. We establish universality and regularity properties.

We run performance experiments using a variety of platforms such as  Intel, AMD and ARM processors, with both Microsoft and GNU compilers.
On recent Intel processors, our proposal (\treename{}) hashes long strings at a rate of   4.7~bytes per cycle for 32-bit outputs 3.3~bytes per cycle for 64-bit outputs.
Generally, our functions match the speed of  state-of-the-art hash functions: they are as fast as  MurmurHash~\cite{smhasher} on shorter segments and comparable  to VHASH~\cite{dai2007vhash} on longer data segments. However, \treename{} has distinguishing theoretical guarantees: MurmurHash is not claimed to be universal and VHASH offers poor regularity.

We also present the optimization methods used to achieve the good performance of  \treename{} (see \S~\ref{ref:efficient}). For example, we optimize the computation of a scalar product in a finite field. Such optimizations might be useful for a variety of functions.

\section{Random hashing}

\begin{table}
\caption{\label{tab:notation} Notation }
\centering
\begin{tabular}{ll} \toprule
$h,f, g, f_i$  & hash functions\\
$\mathcal{H}, \mathcal{F}, \mathcal{G}$& families of hash function\\
$X,Y, Z$& sets of integer values\\
$x\in X$   & value $x$ in $X$\\
$|X|$& cardinality of the set $X$\\
$p$& prime number\\
$\kappa$    & parameter of the \treename{} family\\
$s$   & string\\

$s_i$   & value of the $i^{\mathrm{th}}$~character\\
$m$ & number of characters in a block\\
$n$ & number of bits\\
$L$ & number of levels\\
\bottomrule
\end{tabular}
\end{table}

Good hash functions are such that hash values are  random in some sense. To achieve randomness,
we pick a hash function $h:X \to Y$ at random in a family of
hash functions. (For our notation, see Table~\ref{tab:notation}.) For practical reasons, we assume  $Y$ to be an interval of integers starting at zero, e.g., $Y=[0,2^{32})$.

A family is  \emph{uniform} if
$P(h(x) = c) = 1/|Y| $ for any constant $c\in Y$ and any $x \in X$ where $|Y|$ is the cardinality of $Y$~\cite{dietzfelbinger1996universal}. Uniformity is a weak property: let $\mathcal{H}$ be the family of hash functions of the form $h(x)=c$ for some $c \in Y$, then $\mathcal{H}$ is uniform even though each hash function maps all values to the same constant $c$. If $P(h(x) = c) \leq \epsilon $ for all $x$ and all $c$, then we say that it is $\epsilon$-almost uniform.

A family  is
  \emph{universal}~\cite{carter1979universal,Cormen:2009:IAT:1614191} if
the probability of a collision is no larger than if the hash values were random: $P\left (h(x)=h(x')\right )\leq 1/|Y|$
for any $x,x' \in X$ such that $x \neq x'$.
It is \emph{$\epsilon$-almost universal}~\cite{188765} (also written $\epsilon$-AU) if the probability of a collision is bounded by $\epsilon<1$.
Informally, we say that a family has \emph{good universality} if it is $\epsilon$-almost universal for a small $\epsilon$. Universality does not imply uniformity. 

For example, consider Carter-Wegman polynomial hashing~\cite{carter1979universal}. It is given by the family of functions $h:Y^m\to Y$ of the form $h(s_1, s_2, \ldots, s_m) = \sum_{i=1}^m t^{n-i} s_i$ where the computation is executed over a finite field of cardinality $|Y|$. The value $t$ is picked in $Y$.  It is $(m-1)/|Y|$-almost universal but not uniform (even when $m=1$)~\cite{Lemire2012604}.

\subsection{$\Delta$-universality}

A family is  \emph{$\Delta$-universal}~($\Delta$U)~\cite{stinson1996connections} if
$P(h(x) = h(x') + c \bmod \, |Y|) \leq 1/|Y|$ for any constant $c$
and any $x,x' \in X$ such that $x \neq x'$. Moreover, it is  $\epsilon$-almost $\Delta$-universal ($\epsilon$-A$\Delta$U or $\epsilon$-ADU) if
$P(h(x) = h(x') + c  \bmod \,|Y|) \leq \epsilon$ for any constant $c$ and any $x,x' \in X$ such that $x \neq x'$.
$\Delta$-universality implies universality but not uniformity.

It is necessary sometimes to take an $L$-bit hash value and hash it down to $[0,m)$. It is common to simply apply a modulo operation to achieve the desired result. As long as the original hash family is $\Delta$-universal, the modulo operation is a sound approach  as the next lemma shows.

\begin{lemma}\label{lemma:deltaismarvellous} (Dai and Krovetz~\cite[Lemma~4]{dai2007vhash})
Given an $\epsilon$-almost $\Delta$-universal family $\mathcal{H}$ of hash functions  $h:X \to Y$, the family of hash functions
$\{ h(x) \bmod \, M \given  h \in \mathcal{H}\}$ from $X$ to $[0,M)$ is
$\left \lceil \frac{2|Y|-1}{M} \right \rceil \times \epsilon$-almost $\Delta$-universal. Moreover, if $M$ divides $|Y|$, then the result is an $\frac{|Y|}{M}  \times \epsilon$-almost $\Delta$-universal family.
\end{lemma}
%\begin{proof}
%Consider the equation $h(x)  = h(x') +c \bmod M$ for $x \neq x'$. We have \begin{align*}P( &h(x)    = h(x') +c  \bmod M ) \\
%&   =   \sum_{k \given k M \in [0,|Y|) } P(h(x) = h(x')+c  +k M \bmod |Y| )\end{align*} where the sum is over $k$ values such that $k M$ lies in $[0,|Y|)$. There are exactly $\left \lceil \frac{|Y|}{M} \right \rceil$ such values for $k$: \begin{align*}0, 1, \ldots, \left \lceil \frac{|Y|}{M}\right \rceil M - 1.\end{align*} But by the  $\epsilon$-almost $\Delta$-universality, we have that
%\begin{align*}P(h(x) = h(y) +c + k M \bmod |Y|) \leq \epsilon\end{align*} for any $k$ and any $c$
%so that
%$P(h(x) \bmod M = h(x')+c  \bmod M) \leq  \left \lceil \frac{|Y|}{M} \right \rceil \times \epsilon$.
%\end{proof}

%As a corollary of this result, we have that if $M$ divides $|Y|$, then the result is an $\frac{|Y|}{M}  \times \epsilon$-almost $\Delta$-universal family.
%Our Lemma~\ref{lemma:deltaismarvellous} improves on a similar result by  Dai and Krovetz~\cite[Lemma~4]{dai2007vhash} where they have a collision bound of $\left \lceil \frac{2|Y|-1}{M} \right \rceil$ instead of our tight bound of $\left \lceil \frac{|Y|}{M} \right \rceil$.

Lemma~\ref{lemma:deltaismarvellous} encourages us to
seek low collision probabilities if we expect users
to routinely rely on only a few bits of the hash result.
For example, let us consider Bernstein's~\cite{Bernstein2005} state-of-the-art 128-bit  Poly1305 family.
It is $\epsilon$-almost $\Delta$-universal with $\epsilon=8 \lceil L/16 \rceil /2^{106}$ where $L$ is the size of the input in bytes. For all but very large values of $L$, $\epsilon$ is very small. However,  if we reduce Poly1305 to 32~bits by a modulo operation, the result is  $\epsilon$-almost $\Delta$-universal with  $\epsilon=8 \lceil L/16 \rceil /2^{10}$ by Lemma~\ref{lemma:deltaismarvellous}. In other words, it might be possible to find two 2040-byte strings that always collide on their  first 32 bits when using the Poly1305 hash family. Though this is not a problem in a cryptographic setting where a collision requires all 128~bits to be equal, it can be more of a concern with hash tables.

\subsection{Strong universality}
\label{sec:strong}

A family is \emph{strongly universal}~\cite{wegman1981new} (or pairwise independent) if given  2~distinct values $x,x' \in X$, their hash
values are independent: $
P\left (h(x)= y  \; \land  \; h(x')= y'  \right ) =\frac{1}{|Y|^2}$
for any hash values $y,y' \in Y$.
Strong universality implies uniformity, $\Delta$-universality  and universality.
Intuitively, strong universality means that given $h(x)=y$, we cannot tell anything about the value of $h(x')$ when $x'\neq x$.

When $M$ divides $|Y|$,  if $\mathcal{H}$ is strongly universal then so is $\{ h(x) \bmod \, M \given h \in \mathcal{H}\}$. To put it another way, if $\mathcal{H}$ is  strongly universal with $|Y|$ a power of two, then  selecting the first few bits  preserves strong universality.

The \textsc{Multilinear} hash family is a  strongly universal family~\cite{carter1979universal,Lemire10072013}.
It is the addition of a constant value with the scalar product between random values (sometimes called \emph{keys}) and the input data represented as vectors components ($s_1, \ldots, \allowbreak s_m$), where operations and values are over a finite (or Galois) field: $h(s1, s_2, \ldots, s_m)=a_0+\sum_{i=1}^m a_{i} s_i$.
The hash function $h$ is specified by the  randomly generated values $a_0, a_1, a_2, \dots, a_m$.  In practice, we often pick  a \emph{finite field}  $\mathbb{F}_p$ having prime cardinality ($p$). Computations
in $\mathbb{F}_p$ are easily represented using ordinary integer arithmetic on a computer: values are integers in $[0,p)$, whereas additions and multiplications are followed by a modulo operation ($x \times_{\mathbb{F}_p} y = x y \bmod{p}$ and   $x +_{\mathbb{F}_p} y = x + y \bmod{p}$).

%It may seem as if using $m+1$~random keys to hash a message of length $m$ is inefficient. That is, to hash $64 m$~bits of data, we need at least $64 m$~random bits. However, this is   optimal~\cite{Lemire10072013,188765}. In effect, to have good universality over large data objects, we need to have a correspondingly large set of random bits~\cite{Lemire2012604}.

There are weak versions of strong universality that are stronger than $\epsilon$-almost universality. E.g., we say that the family is $\epsilon$-almost strongly universal if it is uniform and if
\begin{align*}P\left (h(x)= y  \mid h(x')= y'  \right )\leq \epsilon\end{align*} for any distinct $x, x'$. It is  $\epsilon$-variationally universal if it is uniform and if
\begin{align*}\sum_{y\in Y }\big | P\left (h(x)= y  | h(x')= c  \right )-1/|Y| \big | \leq 2 \epsilon\end{align*} for all distinct $x,x'$ and for any~$c$~\cite{krovetz2006variationally}.
There are also stronger versions of strong universality such as $k$-wise  independence~\cite{lemi:one-pass-journal,Lemire2012604}.
For example, Zobrist hashing~\cite{zobrist1970, zobrist1990new,thorup2012tabulation} is 3-wise independent (and therefore strongly universal). It is defined as follows.
Consider the family $\mathcal{F}$ of all possible functions $X \to Y$. There are $|Y|^{|X|}$ such functions, so that they can each be represented using $|X| \log |Y|$~bits.
Given strings of characters from $X$ of length up to $N$, pick $N$~functions from $\mathcal{F}$, $f_1, f_2, \ldots, f_N$ using $N |X| \log |Y|$~bits. The hash function is given by $s \to f_1(s_1) \xor \cdots \xor  f_{|s|}(s_{|s|})$ where $\xor{}$ is the bitwise exclusive or. Though Zobrist hashing offers strong universality, it may require a lot of memory. Setting aside the issue of cache misses, current x64 processors cannot sustain more than two memory loads per cycle which puts an upper bound on the speed of Zobrist hashing. In an exhaustive experimental evaluation of hash-table performance, Richter et al.~\cite{Richter2015} found that Zobrist hashing produces a low throughput. Consequently, the authors declare it to be ``less attractive in practice'' than its strong randomness properties would suggest. 

\subsection{Composition and concatenation of families}

There are two common ways to combine families of hash functions: composition ($h(x)=g\circ f (x) \equiv g(f(x))$) and concatenation ($h(x) = (g(x),f(x))$ or $h=(g,f)$).\footnote{Some authors might refer to a concatenation  as a cartesian product or a juxtaposition.}
For completeness, we review important results found elsewhere~\cite{188765}.
 Under composition uniformity is preserved, but universality tends to degrade linearly in the sense  that the bounds on the collision probability add up (see Lemma~\ref{lemma:deltauniversalcomposition}).

\begin{lemma} \label{lemma:deltauniversalcomposition}
Let $\mathcal{F}$ and $\mathcal{G}$ be $\epsilon_{\mathcal{F}}$-almost and $\epsilon_{\mathcal{G}}$-almost universal families of hash functions $f:X\to Y$ and $g:Y\to Z$.
Let $\mathcal{H}$ be the family of hash functions $h:X\to Z$ made of the functions $h=g \circ f$ where $f \in \mathcal{F}$ and $g \in \mathcal{G}$.
\begin{itemize}
\item Then $\mathcal{H}$ is $\epsilon_{\mathcal{F}} + \epsilon_{\mathcal{G}}$-almost universal.
\item  Moreover, if $\mathcal{G}$ is $\epsilon_{\mathcal{G}}$-almost $\Delta$-universal, then  $\mathcal{H}$ is $\epsilon_{\mathcal{F}} + \epsilon_{\mathcal{G}}$-almost $\Delta$-universal.
\item If $\mathcal{G}$ is uniform then so is $\mathcal{H}$.
\end{itemize}
\end{lemma}

\begin{lemma} \label{lemma:deltauniversalconcatenation}
%Uniformity, universality, strong universality, and $\Delta$-universality
Universality is  preserved under concatenation. That is, let  $\mathcal{F}$ be a family of hash functions
$f:X\to Y$, then  the family made of the  concatenations $(f,f): X \times X \to Y \times Y$ is $\epsilon$-almost universal if  $\mathcal{F}$ is $\epsilon$-almost universal.
%
% Uniformity and universality
% are preserved under concatenation.
% That is, let $\mathcal{F}_{1}$ and $\mathcal{F}_{2}$ be  families of hash functions $f:X_{1}\to Y_{1}$ and $g:X_{2}\to Y_{2}$.
%Let $\mathcal{H}$ be a family of hash functions $h:X_{1} \times X_{2}\to Y_{1} \times Y_{2}$ so that $h(x_{1}, x_{2})=(f(x_{1}), g(x_{2}))$ for some $f \in \mathcal{F}_{1}$ and $g \in \mathcal{F}_{2}$.
%Then, we have the following.
%\begin{itemize}
%\item If $\mathcal{F}_{1}$ and $\mathcal{F}_{2}$  are $\epsilon$-almost universal, so is $\mathcal{H}$.
%%\item If $\mathcal{F}_{1}$ and $\mathcal{F}_{2}$  are $\epsilon$-almost $\Delta$-universal, so is $\mathcal{H}$.
%\item If $\mathcal{F}_{1}$ and $\mathcal{F}_{2}$  are uniform, so is $\mathcal{H}$.
%%\item If $\mathcal{F}_{1}$ and $\mathcal{F}_{2}$  are strongly universal, so is $\mathcal{H}$.
%\end{itemize}
\end{lemma}

\section{Regularity}
\label{sec:regularity}
Though we can require families of hash functions to have desirable properties such as uniformity or universality, we also want individual hash functions to have reasonably good properties. For example, what if a family contains the hash function $h(x)=c$ for some constant $c$? This particular hash function is certainly not desirable! In fact, it is the worst possible hash function for a hash table. Yet we can find many such  hash functions in a family that is otherwise strongly universal.
Indeed,
Dietzfelbinger~\cite{dietzfelbinger1996universal} proposed a strongly
universal family made of the hash functions
\begin{eqnarray*}h_{A,B}(x) = \left (Ax +B \bmod{2^{K}}\right ) \div 2^{n-1}\end{eqnarray*}
with integers $A,B \in [0,2^K)$. It is strongly universal over the domain of integers $x \in [0,2^n)$. However, one  out of $2^K$~hash functions has $A=0$. That is, if you pick a hash function at random, the probability that you have a constant function ($h_{0,B}(x)=B \div 2^{L-1}$) is $1/2^K$. Though this probability might be vanishingly small, many of the other hash functions have also poor distributions of hash values. For example, if one picks $A=2^{K-1}$, then any two hash values ($h_{A,B}(x)$ and $h_{A,B}(x')$) may only differ by one bit, at most.
Letting $A$ be odd also does not solve the problem: e.g., $A=1, B=0$ gives the hash function $x\div 2^{n-1}$ which is either 1 or 0.

Such weak hash functions
are a security risk~\cite{Handschuh2008,Saarinen2012}.
Thus, we require as much as possible that hash functions be \emph{regular}~\cite{bellare2004hash,Canetti:1998:POP:276698.276721}.
\begin{definition}
 A hash function $h:X \to Y$  is regular if for every $y \in Y$, we have that
$|\{x \in X:(h(x) = y)\}| \leq  \lceil |X|/|Y| \rceil $.
Further, a family $\mathcal{H}$ of hash functions is regular if every $h \in \mathcal{H}$ is regular.
\end{definition}

We stress that this regularity property applies to
individual hash functions.\footnote{In contrast, Fleischmann et al.~\cite{Fleischmann:2011:9PN:2403503.2403509} used the term $\epsilon$-almost regular to indicate that a family is almost uniform: $P(h(x)=y) \leq \epsilon$ for all $x$ and $y$ given that $h$ is picked in $\mathcal{H}$.}
However, we can still give a probabilistic interpretation to regularity: if we pick any two values $x_1$ and $x_2$ at random, the probability that they collide $h(x_1)=h(x_2)$ should be minimal ($|Y|/|X|$) if $h$ is regular.

As an example, consider the case where $X=Y=\{0,1\}$. There are only two regular hash functions $h:X \to Y$.
The first one is the identity function ($h_I(0)=0, h_I(1)=1$) and the second one is the negation function ($h_N(0)=1, h_N(1)=0$). The family $\{h_I,h_N\}$ is uniform and universal: the collision probability between distinct values is zero.

More generally, whenever  $X=Y$, a function $h: X \to Y$ is regular if and only if it is a permutation.
This observation suffices to show that  it is not possible to have strong universality and regularity in general. Indeed, suppose that $X=Y$, then all hash functions $h$ must be permutations.
Meanwhile, strong universality means that given that we know the  hash value $y$ of the element $x$ (i.e., $h(x)=y$), we still known nothing about the hash value of $x'$ for $x'\neq x$.
But if $h$ is a permutation,  we know  that the hash values differ  ($h(x')\neq h(x)$)---contradicting strong universality.
More formally, if $h$ is a permutation, we have that  $h(x) \neq h(x')$ for $x\neq x'$ which implies that $P(h(x) = h(x'))=0$ whereas $P(h(x) = h(x'))=1/|Y|$ is required by strong universality.
Thus, while we can have both universality and regularity, we cannot have both strong universality and regularity.

The next two lemmas state that regularity is preserved under composition and concatenation.

\begin{lemma} \label{lemma:xuniformcomposition}
(Composition) Assume that $|Y|$ divides $|X|$ and $|Z|$ divides $|X|$. Let  $f:X\to Y$ and $g:Y\to Z$ be regular hash functions then $f \circ g:X\to Z$ is also regular.
\end{lemma}

\begin{lemma} \label{lemma:xuniformconcatenation}
(Concatenation) Let  $f:X_{1}\to Y_{1}$ and $g:X_{2}\to Y_{2}$ be regular hash functions then the function $h:X_{1} \times X_{2}\to Y_{1} \times Y_{2}$ defined by  $h(x_{1}, x_{2})=(f(x_{1}), g(x_{2}))$ is also regular.
\end{lemma}

\subsection{Component-wise regularity}

We can also consider stronger forms of regularity. Consider hash functions of the form $f: X_1 \times X_2 \times \cdots \times X_m  \to X$, then  the hash function is \emph{component-wise regular}  if we can arbitrarily fix all input components but one and still generate all hash values fairly, that is   \begin{align*}|\{x \in X_i:(h(x_1, \ldots, x_{i-1},x,x_{i+1},\ldots, x_m) = y)\}| \\\leq  \lceil |X_i|/|Y| \rceil \end{align*} for any $i$, any $y$ and any values \begin{align*}x_1, x_2, \ldots, x_{i-1},x_{i+1},\ldots, x_m.\end{align*} Intuitively,  component-wise regularity ensures that if we pick two same-length strings at random that differ in only one pre-determined component, the collision probability is minimized. By inspection, component-wise regularity is preserved under composition and concatenation.

Of particular interest are the hash functions of the form  $f: X \times X \times \cdots \times X  \to X$. In this case, component-wise regularity implies that the restriction of the function to one component (setting all other components to constants) is a permutation of $X$.
Clearly, if $f_1$ and $f_2$ are two such functions then their concatenation ($(f_1,f_2)$) is also component-wise regular, and if $g$ is itself component-wise regular, then the composition of $g$ with the concatenation $(f_1,f_2)$, written $g(f_1,f_2)$, is again component-wise regular. The following lemma formalizes this result.

\begin{lemma} \label{lemma:xuniformcompositionconcatenation}
 Let $f_i:X^m \to X$ be component-wise regular hash functions  for $i=1,\ldots,m$. Let $g:X^m \to X$ be a component-wise regular hash function. Then the composition and concatenation $g(f_1,f_2,\ldots, f_m)$ is component-wise regular.
 \end{lemma}

%%%%%%%%%%%
% Daniel: The notion of length-wise regularity is never actually needed.
% Moreover, it is somewhat difficult to understand out of context like this. That is, the motivation is likely unclear for the reader who has not read the rest of our paper.
% It is better to drop.
%%%%%%%%%%%
%\subsubsection{Length-wise regularity}
%
%Sometimes it may be useful to have a hash function that is not only regular on all its domain, but also on each subset consisting of elements of some particular length. For instance, we may want $f$ to be regular on subsets of strings of length 2. More precisely, let $f$ be a hash function on a set of the form $\{x=(x_{1}, \ldots, x_{n})\}$, then we may want $f$ to be also regular on each subset $\{x = (x_1, \ldots, x_{k},x_{k+1}, \ldots,x_{n})\}$ where $x_{k+1}, \ldots,x_{n}$ are fixed. We  refer to such functions as  \emph{length-wise regular}.
%
%It is clear that a length-wise regular function is regular. Also, if $f$ is component-wise regular, then it is length-wise regular. The converse  is not true in general, that is, length-wise regularity does not imply component-wise regularity. Indeed, consider a set $X = X_1 \times X_2$
%with $|X_1| = |X_2|$, and a function $f(x) = f(x_1, x_2) = g(x_1)$ where $g$ is regular on $X_1$.
% Then $f$ is regular on both $X$ and its subset effectively consisting only of the first component, so $f$ is length-wise regular. At the same time, $f$ is not regular on the second component, and thus it is not component-wise regular.
%

\subsection{$K$-regularity}
Regularity is not always reasonable: for example, regularity implies that $|Y|$ divides $|X|$. Naturally, we can weaken the definition of regularity: we say that hash function is \emph{$K$-regular} if $h(x) = y$ is true for at most $K \lceil |Y| / |X| \rceil $~values $x$ given a fixed $y$. A 1-regular function is simply regular.  We define component-wise $K$-regularity in the obvious manner.
 Our objective is to achieve $K$-regularity for a small value of $K$.

Though regularity is preserved under composition, $K$-regularity is not. Indeed, consider the 4-regular function $h:\{0,1,\ldots,2^n-1\}\to \{0,1,\ldots,2^n-1\}$ given by $h(x)=\lfloor x/4 \rfloor $. Composing $h$ with itself, we get a 16-regular function $h'(x) =h(h(x))  =\lfloor x/16 \rfloor$. The example illustrates the following lemma.

\begin{lemma}
Let $f:X\to Y$ and $g:Y\to Z$ be $K_1$-regular and $K_2$-regular hash functions then $f \circ g:X\to Z$ is $(K_1 \times K_2)$-regular if $|Y|$ divides $|X|$ and $|Y|$ divides $|Z|$.
\end{lemma}
\begin{proof}
Given $z\in Z$, we have that $g(y)=z$ is true for at most $K_2 |Z|/|Y|$ values $y\in Y$. In turn, we have that $h(x)=y$ for at most $K_2 |Y|/|X|$ values $x\in X$.
Thus, given $z\in Z$, we have that $g(f(x))=z$ is true
for at most $K_2 |Z|/|Y| \times K_2 |Y|/|X|= K_1 K_2 |Z|/|X|$, completing the proof.
\end{proof}

Thus, in general, regularity degrades exponentially under composition. In contrast, universality degrades only linearly under composition: an $K_1/2^n$-almost universal family composed with another $K_2/2^n$-almost universal is at least $(K_1+K_2)/2^n$-almost universal (by Lemma~\ref{lemma:deltauniversalcomposition}).

To achieve strong regularity, a good strategy might be to only compose functions that are 1-regular. Of course,
we might still need to reduce the hash values to a useful range.
Thankfully,  regularity merely degrades to 2-regularity  under modulo operations.

\begin{lemma}\label{lemma:regularismarvellous}
Given a regular hash function $h: X \to Y$, we have that the hash function $h'$ defined by \begin{align*}h'(x)= h(x) \bmod \, M \end{align*} for $M\leq |Y|$ is regular if $M$ divides $|Y|$ and 2-regular otherwise.
\end{lemma}
\begin{proof}
Pick any value $y \in [0,M)$. If $h'(x)=y$ then $h(x)= y + k M \bmod{|Y|}$ for some integer $k$ such that $kM \in [0,|Y|)$. There are $\left \lceil \frac{|Y|}{M} \right \rceil$ such values for $k$: \begin{align*}0, 1, \ldots, \left \lceil \frac{|Y|}{M}\right \rceil M - 1.\end{align*} Because $h$ is regular,  the equation $h'(x)=y$ has at most $\left \lceil \frac{|Y|}{M} \right \rceil \times \frac{|X|}{|Y|}$~solutions for $x$.
To determine the $K$-regularity, we have to divide this result by $\frac{|X|}{M}$:
$K=\left \lceil \frac{|Y|}{M} \right \rceil \times \frac{|X|}{|Y|} \times \frac{M}{|X|}=\left \lceil \frac{|Y|}{M} \right \rceil \times\frac{M}{|Y|} \leq \left (\frac{|Y|}{M} + \frac{M-1}{M}\right )\frac{M}{|Y|} = 1 +  \frac{M-1}{|Y|}$. We have that $K\leq 2$ in general and $K=1$ if $M$ divides $|Y|$.
\end{proof}

\section{A tree-based construction for universality and regularity}\label{sec:pyr}

We want to address long objects, such as variable-length strings, while  maintaining the collision probability as low as possible. Though we could get strong universality with the \textsc{Multilinear} hash family, we would need as many random bits as there are bits in our longest object. What if some of our objects use gigabytes or more? It may simply not be practical to generate and store so many random bits. To alleviate this problem, it is common to use a tree-based approach~\cite{703969,wegman1981new,sarkar2011trade}. Such an approach allows us to hash very long strings using hash functions that require only a few kilobytes for their description. It is a standard approach so we present it succinctly.

Let $X$ be a set of integers values containing at least the values 0 and 1. 
We pick $L$~hash functions $f_1, f_2, \ldots : X^m \to X$ from a family $\mathcal{H}$ (e.g., \textsc{Multilinear} family from \S~\ref{sec:strong}).
Take any string $s$ made of $N$~character values from $X$  and let $L= \lceil \log_m N + 1\rceil$. Append the value 1 at the end of the string $s$ to create the new string $\sigma$~\cite{dai2007vhash,Boesgaard2005,krovetz2007message,krovetz2001fast}.
If $L=1$, simply  return $f_1(\sigma)$ with the convention that we pad $\sigma$ with enough zeros that it has length $m$.
 If $L>1$,
split $\sigma$ into $m^{L-1}$~segments of length $m$ each (except for the last segment that might need padding)
and apply $f_L$ on each segment: the result is a
 new string of length at most $m^{L-1}$.
 Split again the result into $m^{L-2}$~segments of length at most $m$ each and apply $f_{L-1}$ on each segment.
Continue until a single value remains. See
Fig.~\ref{fig:pyramidal} for an illustration and see Algorithm~\ref{alg:pyr} for the corresponding pseudocode.

If the family $\mathcal{H}$ is $\epsilon$-almost universal, then the family formed by the tree-based construction has to be $L\epsilon$-almost universal. Indeed, it can be viewed as the composition of $(f_L, f_L ,\ldots )$, $(f_{L-1}, f_{L-1} ,\ldots )$, \ldots, $f_1$. Each one is $\epsilon$-almost universal by Lemma~\ref{lemma:deltauniversalconcatenation}. And the composition of $L$ $\epsilon$-almost universal functions is $L\epsilon$-almost universal by Lemma~\ref{lemma:deltauniversalcomposition}.

Moreover, because $f_1$ is $\epsilon$-almost $\Delta$-universal, and the composition of $(f_L, f_L ,\ldots )$,
$(f_{L-1}, f_{L-1} ,\ldots )$, \ldots, $(f_{2}, f_{2} ,\ldots )$ is $(L-1)\epsilon $-almost universal, we have that the final construction must be $L\epsilon$-almost $\Delta$-universal.
Further, as long as  $\mathcal{H}$ is a uniform family, the construction is uniform.
Moreover, by Lemma~\ref{lemma:xuniformcompositionconcatenation}, we have that if the family $\mathcal{H}$ is regular and component-wise regular, then the result from the construction is regular component-wise regular as well.

\begin{figure*}\centering
\subfloat[String \texttt{ab} \label{fig:pyramidalshort}]{%
\begin{tikzpicture}[<-,>=stealth', text height=0.3cm,
    text centered,
    level 1/.style={level distance=1cm,sibling distance=0.8cm,minimum width=0.5cm}]
    draw/.style={text centered}
  \node [draw] {$f_1(\texttt{a,b},1,0)$}
child{node [draw] {\texttt{a}} }
child{node [draw] {\texttt{b}} }
child{node [draw] {1} }
child{node [draw] {0} }
    ;
\end{tikzpicture}
}
\hspace*{1cm}
\subfloat[String \texttt{abcde}\label{fig:pyramidallong}]{%
\begin{tikzpicture}[<-,>=stealth', text height=0.3cm,
    text centered,
    level 1/.style={level distance=1cm,sibling distance=3.2cm},
    level 2/.style={level distance=1cm,sibling distance=0.8cm,minimum width=0.5cm}]
    draw/.style={text centered}
    \node [draw] (A){$f_2(f_1(\texttt{a,b,c,d}), f_1(\texttt{e},1,0,0), 0 , 0)$}
    child{node [draw] {$f_1(\texttt{a,b,c,d})$}
child{node [draw] {\texttt{a}} }
child{node [draw] {\texttt{b}} }
child{node [draw] {\texttt{c}} }
child{node [draw] {\texttt{d}} }
    }
child{node [draw] {$f_1(\texttt{e},1,0,0)$}
child{node [draw] {\texttt{e}} }
child{node [draw] {$1$} }
child{node [draw] {$0$} }
child{node [draw] {$0$} }
}                  ;
\end{tikzpicture}
}
\caption{\label{fig:pyramidal}Simplified tree-based algorithm  hashing the strings $\texttt{ab}$ and $\texttt{abcde}$ using hash functions $f_1,f_2,\ldots:X^4\to X$ }
\end{figure*}
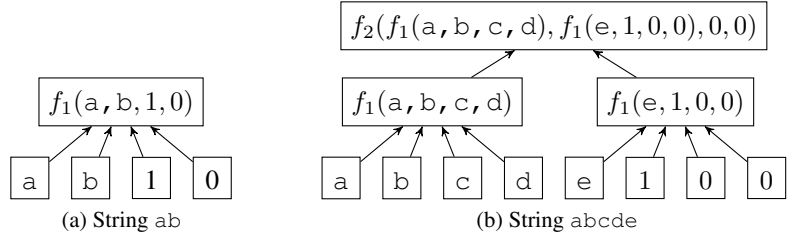

%Fig.~\ref{fig:pyramidal} illustrates a simplified pyramidal hashing algorithm with hash functions $f_1,f_2,\dots:X^4\to X$.
% The strings are first padded by the addition of the character value 1, and then we add zeros up to the nearest multiple of 4.
%In Fig.~\ref{fig:pyramidalshort}, the string \texttt{ab} is hashed using a single hash function ($f_1$). However, for the longer string \texttt{abcde}, two hash functions are required ($f_1, f_2$) as shown by Fig.~\ref{fig:pyramidallong}.
%This last example illustrates that  the value zero is substituted for missing values at higher levels. In this instance, only two levels are required but strings of length greater or equal to $4 \times 4=16$ would require three levels.

%The resulting algorithm (see Algorithm~\ref{alg:pyr} for a simplified form) is $\epsilon L$-almost universal, as required.
%
%\begin{lemma} The hash family described by Algorithm~\ref{alg:pyr}
%is
%$\epsilon L$-almost $\Delta$-universal and uniform.
%\end{lemma}

\begin{algorithm}
\caption{Tree-based algorithm}\label{alg:pyr}
\begin{algorithmic}[1]\small
\REQUIRE Set of integer values $X$ containing at least the values 0 and 1. \COMMENT{E.g., set of all 32-bit integers.}
\REQUIRE $L$ hash functions $f_1, f_2, \ldots, f_{L}$ of the form $X^m \to X$ for $m>1$,  picked independently from a family $\mathcal{H}$ that
is uniform and $\epsilon$-almost $\Delta$-universal.
\STATE \textbf{input}: string $s$ made of $N$~character values from $X$ with $1 \leq N \leq m^L-1$. \COMMENT{That is, $s \in X^N$ and $|s|=N$.}
\STATE Let $\sigma$ be the string of length $N+1$ that we get by appending the value 1 at the end of the string $s$. \COMMENT{We have that $|\sigma|\leq m^L$.}
\STATE $j \leftarrow 1$
\WHILE {$\sigma$ contains more than one character value ($|\sigma|>1$)}
\STATE while the length $|\sigma|$ is not a multiple of $m$, append a zero to $\sigma$.
\STATE $\sigma \leftarrow f_j(\sigma_1, \ldots, \sigma_m),f_j(\sigma_{m+1}, \ldots, \sigma_{2m}), \ldots,$ \\\hspace{1cm}$f_j(\sigma_{|\sigma|-m+1}, \ldots, \sigma_{|\sigma|})$
\STATE $j \leftarrow j+1$
\ENDWHILE
\RETURN the sole character value of $\sigma$ as the hash value of $s$
\end{algorithmic}
\end{algorithm}

%\subsection{Component-wise regularity}
%
%
%We also want component-wise regularity. Consider Algorithm~\ref{alg:pyr} with
%the added requirement that all functions $f_1,  \ldots, \allowbreak f_L$ are
%component-wise regular. This is sufficient to show that the resulting
%hash function is component-wise regular. It suffices to observe that
%component-wise regularity is preserved under composition and concatenation (see Lemma~\ref{lemma:xuniformcompositionconcatenation}). The rest of the analysis follows
%by inspection.
%
%
%\begin{lemma}\label{lemma:cwreg} Under the requirement that the functions $f_1, \allowbreak \ldots, \allowbreak f_L$ are
%component-wise regular,  the family described by Algorithm~\ref{alg:pyr}
%is component-wise regular.
%\end{lemma}

To achieve almost $\Delta$-universality, it is only required that the last of the hash functions applied come from an almost $\Delta$-universal family. Thus it is possible to use families with weaker universalities (e.g., merely $\epsilon$-almost universal) as part of the tree-based construction, while still offering almost $\Delta$-universality in the end. However, with regularity, we cannot as easily substitute potentially weaker hash families: we require that all hash functions being composed be regular.

For clarity, we described Algorithm~\ref{alg:pyr} in such a way that the first level is computed entirely as a first step (using $f_1$), followed by a second pass at the second level (using $f_2$) and so on. This approach requires allocating dynamically a possibly large amount of memory. We compute the same result using a bounded and small amount of memory~\cite{Boesgaard2005}: no more than $m(L-1)$~values from $X$.  We first hash the first $m$~characters of the string that has been extended with an extra 1. The result is written at the first location in the second level. We repeat with the next $m$~elements. (Cases where we have fewer than $m$~characters left are also handled efficiently, avoiding copies and explicit zero-padding.) Once we have $m$~hash values stored in the second level, we hash them and store the result in the third level.
After each chunk of $m$~characters is hashed, we push its hash value
to a higher level.  Once we are done hashing the input, we complete the computation.

Almost all data objects in modern computing can be represented as a string of bytes (8-bit words) so we assume that we accept strings of bytes for complete generality. Yet on 32-bit or 64-bit processors, it is not always desirable to process the inputs byte-by-byte: it is more natural and faster to process the data using 32-bit or 64-bit machine words. So our set of characters $X$ is made of all 32-bit or all 64-bit values. When appending the string with a value of 1 as in Algorithm~\ref{alg:pyr}, we actually pad with a 1-byte and zeros to the nearest machine word boundary. In software, we avoid creating a new extended  string with padded bytes---as it would be inefficient. Instead we just compute the final machine word  and use an optimized code path.

%\subsection{Alternative to the pyramidal approach}

As pointed out by Halevi and Hugo~\cite{MMH1997},
there is a downside to the tree-based approach:
the universality degrades linearly with the height of the tree.
We could solve this problem by hashing all but the last level of the tree to a larger domain (e.g., one of cardinality $L|X|$), as long as we could maintain regularity.
Or, instead, we could use a two-level approach where only the first level uses \textsc{Multilinear}, while the second level uses a polynomial hash family: VHASH described in Appendix~\ref{appendix:related} uses a similar approach~\cite{dai2007vhash}. We would need to ensure that we have good regularity in both levels.
However,
we can alleviate this degraded universality problem by using a tree of small height. That is, if we  choose the family $\mathcal{H}$ of hash functions $h:X^m\to X$ with a relatively large integer $m$, we may never require a tall tree (e.g., one with more than $\approx 8$~levels). In this manner,  the tree-based approach may still meet our goals by ensuring component-wise regularity while still achieving good universality.

%Another downside of the pyramidal approach is a slightly increased memory usage. For example,
%that each level requires a set of random keys.
%Moreover, intermediate results in various levels need to be stored during the computation. However, we show in \S~\ref{sec:singlepass} how we can alleviate these problems and use only about 22\,kB to hash very long 64-bit strings. Such a small amount of memory typically fits in the L1 cache of a processor. For most strings, the memory usage is even lower because only the first few levels of the pyramid are used.

\section{Universality and regularity  with \treename{}}

To implement Algorithm~\ref{alg:pyr}, we need to select
a family of hash functions $\mathcal{H}$.
The \textsc{Multilinear} family (see \S~\ref{sec:strong})
might fit our needs in mathematical terms since it is strongly universal:
$h(s)=a_0+\sum_{i=1}^m a_{i} s_i \bmod \, {p}$ for $p$ prime.
By picking keys $a_1, a_2, \ldots $ as integers in $[1,p)$,
we get a component-wise regular and almost universal family.
However, the resulting hash family depends
 crucially on the choice of a prime number $p$.
In related work, authors chose prime numbers smaller than a power of two~\cite{krovetz2007message,dai2007vhash,thorup2012tabulation}  such as Mersenne primes or pseudo-Mersenne primes (primes of the form ${2^n \allowbreak - \allowbreak  k}$ where $k$ is much smaller than $2^n$ in absolute value~\cite{Nussbaumer:1976:DFU:1664374.1664380}).
Such prime numbers enable fast modulo reduction algorithms.
For example, $p=2^{61}-1$ is a Mersenne prime.
Given a 64-bit integer $x$, we can compute $x \bmod \,p$ by first computing $(x \bmod \, 2^{61}) + (x \div 2^{61})$ and then subtracting $p$ from $x$ if $x$ exceeds $p$.

Of course, we do not hash strings of numbers in $[0,p)$ for $p$ prime, instead we hash strings of numbers in $[0,2^n)$.
Choosing $p<2^n$ is not a problem to get almost universality~\cite[Section~4]{krovetz2001fast}. However,
it makes it more difficult to achieve regularity.
%To compensate, when $p<2^n$, we can either map numbers in $[p-1,2^n)$ to numbers in $[0,p-1)$, or deduct from numbers in $[p,2^n)$ a randomly chosen number~\cite{krovetz2001fast}: this extra processing might make it more difficult to achieve regularity.
To illustrate the problem, consider once more
the family $h(x) = a x \bmod \, p$ for $p$ prime and an integer $a$ picked randomly in $[1,p)$. This family is regular for inputs in $[0,p)$. Suppose however that $x \in [0,2^n)$ for $2^n >p$, then the result is at most 2-regular. Because  regularity degrades exponentially with composition, if we use a 2-regular function at each level in the a tree-based setting, the final result might only be $2^L$-regular for trees of height $L$. However, the problem goes away if we pick $2^n < p$, as  $h(x) = a x \bmod \, p$ is then regular once more.

Hence,
our selection of prime numbers ${p}$ is based on \emph{two} requirements: \begin{enumerate}
\item  for a number ${x}$ that fits a single processor word, ${x  \bmod \, p}$ should be equal to ${x}$, thus making it easy to achieve regularity,
 and
\item  reduction modulo ${p}$ of numbers that do not fit to a single word should be expressed in terms of computationally inexpensive operations. In practice, this may be achieved by choosing $p$ close to a power of two matching the processor word size (such as $2^{64}$).
\end{enumerate}
Thus, we use minimal primes that are greater than any number that fits in a single processor word, that is, for a 32-bit platform, ${p = 2^{32} + 15}$, and for 64-bit platform, ${p = 2^{64} + 13}$ (see Table~\ref{table:smallestprimes}).   We call primes  of a form ${2^n+k}$ where $k$ is small \emph{\textsc{pseudo}+\textsc{Mersenne} primes} by analogy with pseudo-Mersenne primes. Table~\ref{table:smallestprimes} gives several such
primes, e.g., $2^{64}+13$.

The idea of using \textsc{pseudo+Mersenne} primes for universal hashing is not new~\cite{Bernstein2005,MMH1997}.
\begin{itemize}
\item Our approach is similar to  Multidimensional-Modular-Hashing (MMH)~\cite{MMH1997}. The MMH authors use $p=2^{32}+15$ for $n=32$. They build their hash family on multilinear functions of the form
$h(s)=(\sum_{i=1}^m a_{i} s_i \bmod \, 2^{2n}) \bmod{p}$
(as opposed to
$h(s)=\sum_{i=1}^m a_{i} s_i  \bmod{p}$). That is, they use only two $n$-bit words to compute the sum although more than $2n$~bits are required (e.g., 3~words) to compute the exact sum.
 They prove that their speed
optimization only degrades the universality slightly (by a factor of 2). However, they also degrade the regularity. Because the regularity degrades exponentially with composition in the worst case (see \S~\ref{sec:regularity}), we prefer to avoid non-regular functions for a tree-based construction. Moreover,
we are able to produce fast code to compute the exact sum (see Appendix~\ref{appendix:opttech}). Since we benchmark our contributed functions against a family faster than MMH (VHASH~\cite{krovetz2007message,dai2007vhash}), we do not consider MMH further.
\item Our approach is also related to Bernstein's~\cite{Bernstein2005} cryptographic Poly1305 function that uses $p=2^{130}-5$ to generate 128-bit hash values.
Bernstein reports choosing $p=2^{130}-5$ instead of a value closer to $2^{128}$ for computational convenience. Though it is possible that larger primes than the ones we choose could allow further speed optimizations, it may also degrade the universality slightly.  Thus we do not consider the possibility further.
\end{itemize}

\begin{table}
\caption{\label{table:smallestprimes}Smallest primes larger than a power of two~\cite{A132198}}
\centering \begin{tabular}{cc}
\toprule
Power of two & Smallest prime  \\
\midrule
$2^8$ & $2^{8}+1$\\
$2^{16}$ & $2^{16}+1$\\
$2^{32}$ & $2^{32}+15$\\
$2^{64}$ & $2^{64}+13$\\
$2^{128}$ & $2^{128}+51$\\
\bottomrule
\end{tabular}
\end{table}

From this family of prime numbers, we define the  \nameofthescheme{} family of hash functions.

\begin{definition}
Let  $p$ be a prime,  $X=[0,p)$ and $m$ a positive integer. Let $2^n$ be the largest power of two smaller than $p$.
The \fullnameofthescheme{} family (or \nameofthescheme{}) is the set of functions from $X^m$ to $X$ of the form
\begin{align}\label{eq:2wlinear}f(s) \equiv f(s_1, \ldots,s_m) = \left ( b + \sum_{i=1}^m a_i s_i \right ) \bmod \, p\end{align} where $b$ is an integer in $[0,2^n)$   whereas
 $a_1, \ldots, a_m$ are non-zero integers in $(0,p-\kappa)$ for some integer $\kappa \geq 0$.
\end{definition}

%
%
%\subsection{Analysis of the \nameofthescheme{} family}
%
Observe that integers subject to  additions and multiplications modulo $p$ for $p$ prime form
a \emph{finite field}  $\mathbb{F}_p$. Thus we have that $a x  \bmod \, p = a' x  \bmod \, p$ implies $a = a'$ unless $x= 0$ since $x$ is invertible (in $\mathbb{F}_p$).
 We have that \nameofthescheme{} is component-wise regular: we can solve
 the equation $f(s) = y$ for $s_i$ with exactly one value:
 $s_i = a_i^{-1}(y -  a_1 s_1 - a_2 s_2 - \cdots - a_{i-1} s_{i-1}- a_{i+1} s_{i+1} - a_m s_m )  $ in $\mathbb{F}_p$.
 That is, it suffices to require that the parameters  $a_1, \ldots, a_m$ are non-zero to get regularity.

 We can also show $\epsilon$-almost $\Delta$-universality as follows.
 Consider the equation
 \begin{align*}
 \left  ( b + \sum_{i=1}^m a_i s_i \right ) - \left  ( b + \sum_{i=1}^m a_i s'_i \right  ) \bmod \, p = y
 \end{align*}
 for some $y$ in $[0,p)$ for two distinct strings $s$ and $s'$. We have
 that $s_r \neq s'_r$ for some index $r$. Thus, fixing all other values,
 we can solve for exactly one value $a_r$ such that the equality holds. When picking the hash function at random, $a_r$ can have one of $p-\kappa-1$ different values (all integers in $[1,p-\kappa)$), thus the equation holds with probability at most $1/(p-1-\kappa)$.

 Similarly, we can show that  \nameofthescheme{} is $1/2^n$-almost uniform. Indeed, consider the equation
\begin{align*}
 \left  ( b + \sum_{i=1}^m a_i s_i \right ) \bmod \, p = y
 \end{align*}
  for some $y$ in $[0,p)$.
 Fixing the $a_i$'s, the $s_i$'s and $y$, there is exactly one value $b \in [0,p)$ solving this equation. Yet we have $2^n$ possible values for $b$, hence the result.

 We have the following lemma.
 \begin{lemma}
 The family \nameofthescheme{} is $1/(p-1-\kappa)$-almost $\Delta$-universal, $1/2^n$-almost uniform and component-wise  regular.
 \end{lemma}
 Though it may seem that the parameter $\kappa$ is superfluous as setting $\kappa=0$ optimizes universality, we shall see that restricting the range of values with $\kappa > 0$ can ease computations. Similarly, it may seem wasteful to pick $b\in [0, 2^n)$ instead of picking it in $[0,p)$, but this is again done for computational convenience.

In what follows, we  call  \treename{}
 the use Algorithm~\ref{alg:pyr} with the hash family \nameofthescheme{}.
The result is a hash family that is $L/(p-1-\kappa)$-almost  $\Delta$-universal, uniform and component-wise regular over strings of length up to $m^L-1$.

Naturally, we want the resulting hash values to fit in a more convenient range than $[0,p)$. So, we
compute $h(s) \bmod \, 2^n$. We  have that $\left \lceil \frac{p}{2^n} \right \rceil = 2$.
As per Lemmas~\ref{lemma:deltaismarvellous} and~\ref{lemma:regularismarvellous}, the result is $3L/(p-1-\kappa)$-almost $\Delta$-universal and $2$-regular.\footnote{Since
$\frac{2p-1}{2^n}<3$.}

%\begin{remark*} The \treename{} family is only
%applicable in a tree-based setting when $m>1$. When $m=1$,
%the functions take the form
%$f(x) = ( a x + b ) \bmod p$. Though useless in a tree-based setting, we can use such functions for hashing  integers in $[0,2^n)$ quickly.
%\end{remark*}

%\subsection{Inputs as a string of $n$-bit words and regularity}
%
%
%
%Working with input strings of values in $[0,p)$ for $p>2$ prime would be inconvenient in a binary computer. Thus, we apply hash functions from \nameofthescheme{} to strings of integer values in $[0,2^n)$ where $2^n$ is the largest power of two smaller than $p$. This choice does not affect our results regarding almost $\Delta$-universality and uniformity.
%
%Regularity is not affected in the following sense. The \nameofthescheme{} family is component-wise regular from $[0,p) \times \cdots \times  [0,p) \to [0,p)$. This means that given a string, if we fix all but one component, we get a permutation of $[0,p)$. If further restrict the initial components to $[0,2^n)$, the function is still a one-to-one map. As the final step of our computations,
%we return $h(s) \bmod 2^n$. It is easy to see that the result is, at worse, a two-to-one map, hence we have component-wise 2-regularity.

\section{An efficient implementation of \treename{}}
\label{ref:efficient}

The functions in the \nameofthescheme{} family make use of a modulo operation.
On most modern processors, division and modulo operations are computationally expensive in comparison to addition, or even multiplication.
%For example, on recent Intel processors with the Haswell microarchitecture, the reciprocal throughput of a division instruction is at least 8~cycles as opposed to 1~cycle for the multiplication; similarly, the latency of a division exceeds 20~cycles whereas multiplications have a latency of 3~cycles~\cite{fog2014instruction}. The performance problem becomes even greater if such computations are to be performed over integers with size greater than a single processor word (e.g., with integers greater than ${2^{64}}$ on 64-bit platforms).
Thankfully, equation~\ref{eq:2wlinear} suggests a single modulo operation after a series of multiplications and additions that reduces the number of modulo operations to one per $m$~multiplications~\cite{MMH1997,krovetz2001fast}, where $m$ is a parameter of our family.
%Still, this is not enough for a good performance, since a naive implementation of \nameofthescheme{} would need to use some large number arithmetic library, for instance,  the \texttt{boost::\allowbreak{}multiprecision} library to cope with integers larger than a processor word.
We can furthermore tune \treename{}  by optimizing the scalar product computations   (see \S~\ref{sec:fastmul}) and replacing expensive modulo operation by a specialized routine (see \S~\ref{sec:efficientreduction}).
%  According to our tests, our optimized version is 20 to 40~times faster than a naive implementation.

\subsection{Scalar product computation}
\label{sec:fastmul}
Our data inputs are strings of $n$-bit characters.
Two cases are important: $2^n=2^{32}$ (particularly for 32-bit architectures) and $2^n=2^{64}$ (mostly for 64-bit architectures). Where applicable we refer to them separately as  \nameoftheschemesingle{}  (or \treenamesingle{} in the tree-based version) and  \nameoftheschemedbl{} (or \treenamedbl{} in the tree-based version).

Consider the scalar product between keys and components $\sum_i a_i s_i$. Recall that we pick the values $a_i$ in $(0,p-\kappa)$. As long as we choose $\kappa$ large enough so that $p-\kappa \leq 2^n$, we have that $a_i$ is also a machine-sized word (e.g., 64~bits on a 64-bit platform).
For long data segments, we expect most of the running time to be due to the first level of the tree, and mostly due to the computation of the sum  $\sum_{i=1}^m a_i s_i$.
We describe our fast implementation of such computations on modern superscalar processors in Appendix~\ref{appendix:opttech}.

Though we hash strings of machine-sized words (so that $s_i\in [0,2^n)$), in a tree-based setting (see Algorithm~\ref{alg:pyr}), we can no longer assume that $s_i$ fits in a single word---beyond the first level. Two words are required in general. The $s_i$'s are in $[0,2^{32}+15)$ for \treenamesingle{} and in $[0,2^{ 64}+13)$ for \treenamedbl{} at all but the first level in Algorithm~\ref{alg:pyr}.
(We could reduce the hash values so that they fit in a single word, but it would degrade the regularity and universality of the result.)
For speed and convenience, we still want the result of the multiplication to fit in two words.  That is, we want that $a_i s_i \in [0,2^{2n})$ or, more specifically,
\begin{align*}(p - \kappa) (p-1) < 2^{2n}.\end{align*} For this purpose, we set $\kappa = 24$ for \treenamedbl{}. That is, we pick the $a_i$'s in $(0,2^{64}+13-24)=(0,2^{64}-11)$. For \treenamesingle{}, we set $\kappa = 28$ and pick $a_i$'s in $(0,2^{32}+15-28)=(0,2^{64}-13)$. See Table~\ref{table:parameters}
for the parameters and Table~\ref{table:properties} for the properties of the resulting hash families.

For both the
\treenamesingle{} and \treenamedbl{} cases, we use a maximum of 8~levels ($L=8$). Yet we are unlikely to use that many levels in practice: e.g., if we assume that inputs fit in four gigabytes, then 4~levels are sufficient.
%Currently, most commodity 64-bit processors (ARM and x64) are  limited to 48~bits of addressable memory ($2^{48}$~bytes)
%while 32-bit processors have have generally no more than 32~bits of addressable memory. Thus, in the worst case, we would use 6~levels in  64-bit systems and 4~levels in 32-bit systems.
%In this sense, we are being conservative.

\begin{table}
\caption{\label{table:parameters}Parameters used by \treename{}}
\centering \begin{tabular}{ccccc|c}
\toprule
Word size &$2^n$ & $p$  & $\kappa$ & $m$ & $L$  \\
\midrule
32 bits & $2^{32}$& $2^{32}+15$  & 28 & $128$ & 8 \\
64 bits & $2^{64}$& $2^{64}+13$  & 24 & $128$ & 8 \\
\bottomrule
\end{tabular}
\end{table}

\begin{table*}
\caption{\label{table:properties}Properties of Algorithm~\ref{alg:pyr} applied with \nameofthescheme{} in a tree-based setting (\treename{}). String lengths are expressed in machine words (32~bits or 64~bits).  }
\centering \scriptsize \begin{tabular}{cccccc}
\toprule
Name &Word size & Hash interval & max.\ string length & universality & regularity\\
\midrule
\treenamesingle{}& 32 bits &$[0,2^{32})$ &  $(2^{56}-1)$ words& $\frac{12}{2^{31}-7}$-A$\Delta$U& component-wise 2-regular\\
\\%\foonote{}
\treenamedbl{}&64 bits & $[0,2^{64})$ &  $(2^{56}-1)$ words& $\frac{12}{2^{63}-6}$-A$\Delta$U& component-wise 2-regular\\
\bottomrule
\end{tabular}
\end{table*}

\subsection{Efficient modulo reduction}
\label{sec:efficientreduction}
We have insured that the result of our multiplications fit in two words. Halevi and Krawczyk~\cite[Section~3.1]{MMH1997} have derived an efficient modulo reduction in such cases. However, the sum of our multiplications requires more than two words since, unlike Halevi and Krawczyk, we compute an exact sum.
Thus we need to derive an efficient routine to apply the modulo operation on three input words.
We have that \begin{align*}S=\left (b + \sum_{i=1}^m a_i s_i \right )\leq (2^n -1) + m (2^{2n}-1)\end{align*}
Because we choose $m=128$, we have that the result fits  into three words ${(w_0, w_1, w_2)}$ (either 32-bit or 64-bit words) with a small value stored in the most significant word ${w_2}$ (no larger than $m$).

The modulo reduction uses the equalities (modulo $2^n+k$): ${ 2^n  = -k}$ and ${( 2^n)^2 = (-k)^2 = k^2}$.
In our case, $k=15$ or $k=13$ depending on whether we use a 32-bit or 64-bit platform.
 Then, for any number $S$ that fits in 3~words $w_0,w_1,w_2$ (i.e., it is smaller than $2^{3n}$ where $n=32$ or $n=64$), we have
\begin{align*} S & \equiv  w_0 + w_1 \times 2^n + w_2 \times (2^n)^2 \\  &\equiv   w_0 - k \times w_1 + k^2 \times w_2\end{align*}
modulo $2^n+k$.
Let $u_0 = (k\times w_1 ) \bmod{2^n}$ and $u_1 = (k\times w_1) \div 2^n$, then $k\times w_1 = u_1\times 2^n + u_0$. By substitution, we further obtain
\begin{align*}S & \equiv    w_0 + k^2 \times  w_2 - 2^n \times  u_1  - u_0\\ & \equiv   w_0 + k^2 \times  w_2  + k \times  u_1 - u_0
\\&\equiv ( w_0 + k^2 \times  w_2 + k \times  u_1) + ( 2^n + k - u_0)
\end{align*} modulo $2^n+k$.

We have thus reduced $S$ to a number smaller than $2^{2n}$ (modulo $2^n+k$) which fits in two $n$-bit words. We can therefore write $S\equiv v_0 + 2^n v_1 (\bmod \, {2^n + k})$ where $v_0, v_1 \in [0,2^n)$ are easily computed.

We can also bound our representation of $S$ as follows:
\begin{itemize}
\item $w_0 \leq 2^{n}-1$;
\item $k^2 w_2 \leq k^2 m \leq 15^2 \times 128 = 28800$
\item $k \times u_1  = k \times ( k\times w_1 \div 2^n)
\leq k \times (( k\times 2^{n}-1) \div 2^n)
\leq k (k-1) \leq 210 $;
\item $2^n +k -u_0 \leq 2^n +k -1 \leq 2^n +14$.
\end{itemize}
Thus we have a bound of $2\times 2^n + 29023$. It follows that $v_1 \leq 2$.

To reduce $S$ to a number in $[0,p)$ requires branching (see Algorithm~\ref{alg:reduction}). The algorithm works as follows:
\begin{itemize}
\item If $k \times  v_1 \leq v_0$, we exploit the fact that, modulo $2^n+k$, we have $2^n \times v_1 = -k \times  v_1$ to return $ v_0 - k \times  v_1$.

To accelerate this case in software, we can use the fact that $v_0 \geq 2k \Rightarrow k \times  v_1 \leq v_0$. Since $v_0 \geq 2k$ is common and faster than checking that   $k \times  v_1 \leq v_0$, it is worth introducing an extra branch.
\item  If $k \times  v_1 > v_0$, then we know that $v_1>0$. If $v_1=1$, then $v_0<k$ and we can return $v_0 +2^n$ without any reduction. Otherwise we have that $v_1=2$. In such a case, we use the fact that $2\times 2^n = 2^n -k \bmod{(2^n +k)}$
to write $  v_0 + 2\times 2^n$ as $2^n -k  + v_0$. This value is  smaller than $2^n +k$ since  $v_0 < 2k$.
\end{itemize}

\begin{algorithm}
\caption{Reduction algorithm: find the integer $z \in [0,2^n+k)$ such that  $v_1 +2^n v_2 = z \bmod\, {2^n+k}$.
}\label{alg:reduction}
\begin{algorithmic}[1]\small
\STATE \textbf{input}: an integer $v_1\in [0,2^n)$ and an integer $v_2 \in \{0,1,2\}$ \COMMENT{Represents $v_1 +2^n v_2$.}
\IF{$k \times  v_1 \leq v_0$}
\RETURN $ v_0 - k \times  v_1$ \hfill\COMMENT{Can use $v_0 \geq 2k \Rightarrow k \times  v_1 \leq v_0$ to accelerate the check}
\ENDIF
\IF{$v_1=1$}
\RETURN $v_0+2^n $
\ENDIF
\RETURN $v_0 - k$ \hfill\COMMENT{$v_1=2$ in this case}
\end{algorithmic}
\end{algorithm}

\section{Achieving the avalanche effect}
\label{sec:improvstats}

It is often viewed as desirable
that a small change in the input should lead to a large change in the hash value. For example, we often check whether hash functions satisfy the \emph{avalanche effect}: changing a single bit of the input should flip roughly half the bits of the output~\cite{estebanez2006evolving}.

To improve our hash functions in such respect, we add an extra step to further \emph{mix} the output bits.
We borrowed these procedures from MurmurHash~\cite{smhasher}. For \treenamedbl{} the step in C is
\begin{lstlisting}
 z = z ^ (z >> 33);
 z = z * 0xc4ceb9fe1a85ec53;
 z = z ^ (z >> 33);
\end{lstlisting}
and for \treenamesingle{} it is
\begin{lstlisting}
 z = z ^ (z >> 13);
 z = z * 0xab3be54f;
 z = z ^ (z >> 16);
\end{lstlisting}
where $\xor{}$ is the bitwise exclusive or and $z$ is an unsigned integer of a respective size (32 bit for \treenamesingle{} and 64 bit for \treenamedbl{}).
%Fig.~\ref{fig:bitmixing} illustrates the effect of these functions with an example. We start with one bit sets to 1, after the  first step ($z \leftarrow z \xor{} (z\gg 13)$) two bits are set to 1. The multiplication ($\texttt{0xab3be54f}\times z $) sets to 1 half of the most significant 14~bits, leaving the least significant 18~bits with value 0. The last step ($z \xor{} (z\gg 16)$) leaves us with an apparently random mix of 0s and 1s throughout.
These transformations are invertible for all integers that fit a single word and, therefore, they do not affect universality and regularity.

\section{Experiments}

We implemented \treename{} for the x64, x86 and ARM platforms in C++.
On the x86 platform, we use the SSE2 instruction set for best speed.
We make our software freely available under an open source license.\footnote{\url{http://sourceforge.net/projects/hasher2/}}

To test the practical fitness of the \treename{} schema,  we chose the SMHasher~\cite{smhasher} framework. It provides a variety of performance tests as well as several statistical tests.
For comparison purposes, we  used the same framework to test other hashes commonly used in industry.
\begin{enumerate}
\item The hash functions used in the C++ standard (\texttt{std}) library.
\item The hash functions  used in the Boost library, a widely used C++ library.
\item  MurmurHash~3A for 32-bit platforms and MurmurHash~3F for 64-bit platforms: a popular family of hash functions used by major projects such Apache Hadoop and Apache Cassandra.
\item VHASH~\cite{dai2007vhash,krovetz2007message} (see Appendix~\ref{appendix:related}), one of the fastest hash families on 64-bit processors.
\item SipHash~\cite{aumasson2012siphash}: the 64-bit family  hash functions used by the Python language.
\end{enumerate}
Of course, there are many more fast hash functions (e.g., xxHash, CityHash~\cite{cityhash,So:2012:TFN:2396556.2396575}, SpookyHash~\cite{spooky}, FarmHash~\cite{farmhash}, \textsc{CLHASH}~\cite{clhashing} and tabulation-based or Zobrist hashing~\cite{zobrist1970, zobrist1990new,thorup2012tabulation}).
For a recent review of non-cryptographic hash functions, we refer the interested reader to  Ahmad and Younis~\cite{Ahmad2014197},
Estébanez et al.~\cite{SPE:SPE2179} or Thorup~\cite{thorup2015high}.
We leave a more detailed comparison to future work.

We performed extra steps to ensure these functions work inside the SMHasher testing environment:

\begin{description}
\item[Standard library] To hash an arbitrary length data segment we used the library function \texttt{hash\allowbreak{}$<$string$>$}: it takes a \texttt{string} object as a single parameter. The C++11 standard does not specify the implementation so it is  vendor and even version specific. In practice, the hash value generated occupies 32~bits on 32-bit platforms and 64~bits on 64-bit platforms.

In the context of the SMHasher testing environment, we must first create a string object based on the  data segment and its length to use this function.
To exclude the time used for the creation of the \texttt{std::string} object, we  added a separate method that does just the creation of the object itself, and nothing else. Hence, we were able to  estimate the time required to create the object and deduct it from  the whole processing time. We have observed that time spent on object preparation was roughly 10\,\% of the total processing time.

\item[Boost] Boost is a well regarded C++ library and it is likely that its hash functions are in common use. We  tested the  \texttt{hash\_range( char*, char*)} function from version 1.5 of the Boost library. Like the standard library, the hash value generated occupies 32~bits on 32-bit platforms and 64~bits on 64-bit platforms.

\item[VHASH] We chose to compare against VHASH because it is one of the fastest
families of hash functions on longer data sets: e.g., it is several times faster than high performance alternatives such as Poly1305~\cite{krovetz2007message}. It is faster than UMAC~\cite{krovetz2007message} which has itself found to be twice as fast as MMH~\cite{703969}.
We used the most recent VHASH implementation made available by its authors~\cite{vhashimpl}.
It generates 32-bit hash values on 32-bit platforms whereas it generates 64-bit hash values on 64-bit platforms.
On 64-bit platforms, it  pads data with zeros if the input size is not a multiple of 16~bytes. Such padding results in copying  up to 127~bytes to an intermediate buffer. This operation is done at most once per data segment, and, therefore, affects only relatively short segments. We believe that certain changes in the base implementation might potentially be more efficient than our approach with copying; correspondingly, we calculate and present optimistic estimates that do not include time spent on additional processing (similar to our approach with the standard library).
We proceed similarly on the 32-bit ARM platform.

For the Intel 32-bit platform (x86), the authors' implementation~\cite{vhashimpl} provides two options: one uses SSE2 instructions and another one is in pure C. The performance of the SSE2 implementation  is more than two times higher, so the SSE2 option is used for testing.
This particular SSE2 implementation does not require a particular memory alignment. However, it also assumes that data is processed in blocks of 16~bytes, so we use buffering as in the 64-bit platform.

\end{description}

\subsection{Variable-length results on recent Intel processors}
\label{sec:modern}
In Fig.~\ref{fig:rawspeed}, we compare directly the speeds in bytes per CPU cycle of our hash families over random strings of various lengths. We use a recent Intel processor with the recent Haswell microarchitecture: an Intel i7-4770 processor running at 3.4\,GHz. This processor has  32\,kB of L1 cache per core, 256\,kB of L2 cache per core and 8\,MB of L3 cache. The software was compiled with GNU GCC 4.8 to a 64-bit Linux executable.

In this test, the 64-bit VHASH is capable of hashing 3.9~input bytes per cycle for long strings (4\,kB or more).
\treenamedbl{} is 15\,\% slower on such long strings at 3.3~bytes per cycle.

Our \treenamesingle{} can be 40\,\% faster than \treenamedbl{}, reaching speeds of 4.7~bytes per cycle on long strings. Thus, if we only need 32-bit hash values, it could be preferable to use \treenamesingle{}.
 The speed of MurmurHash~3A is disappointing at 0.8~bytes per cycle, whereas MurmurHash~3F does better at 2~bytes per cycle on long strings. SipHash reaches a speed of 0.5~bytes per cycle. The Boost and std hash functions are slower on long strings (less than 0.25~bytes per cycle).

On short strings, \treenamedbl{} and \treenamesingle{} are fastest followed by MurmurHash~3F.

\begin{figure*}\centering

\subfloat[64 or more bit outputs] {%
\includegraphics[width=0.8\textwidth]{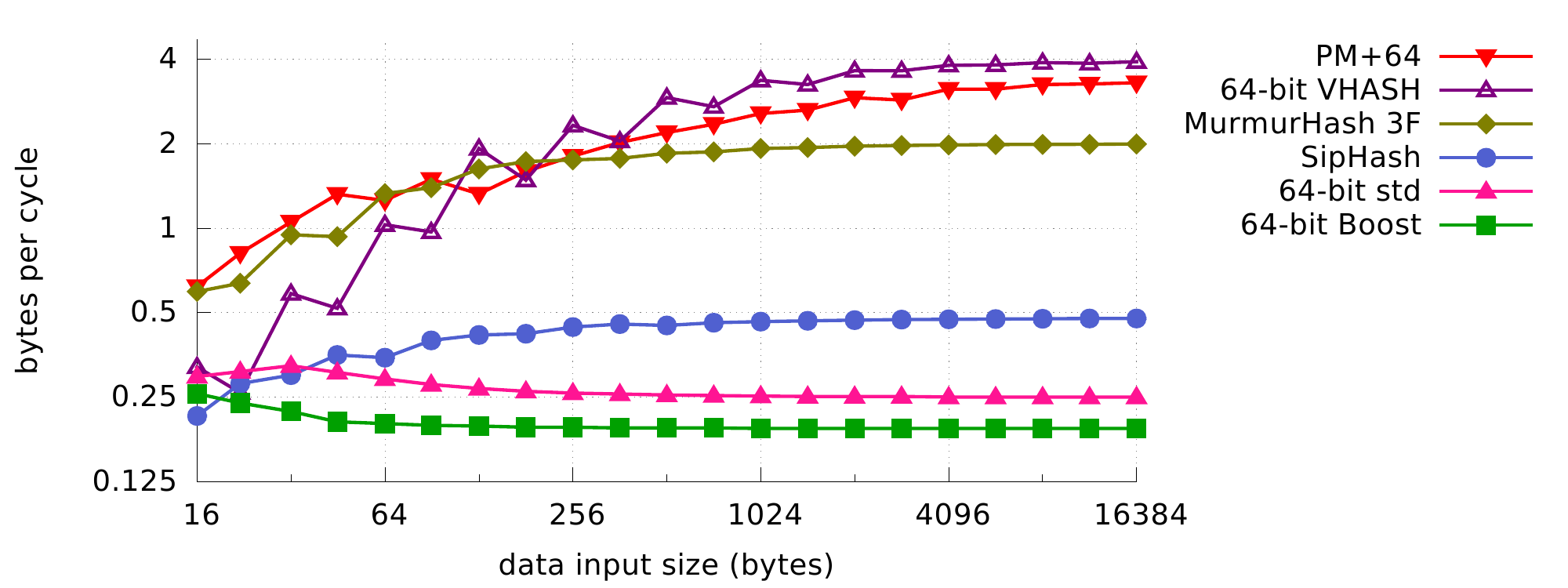}
}

\subfloat[32-bit outputs] {%
\includegraphics[width=0.8\textwidth]{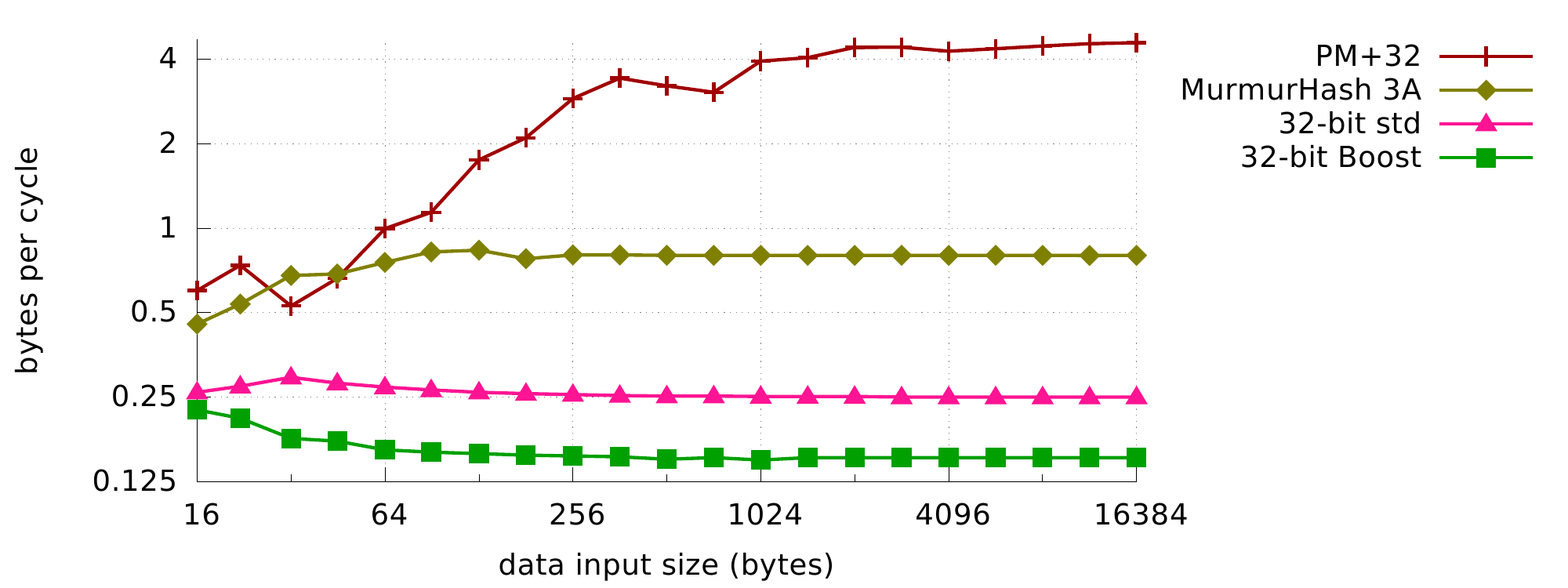}
}
\caption{\label{fig:rawspeed}Speed of the hash functions on random strings of various lengths on the recent Intel  Haswell microarchitecture.}
\end{figure*}

\subsection{Multiplatform performance testing}
\subsubsection{Methodology}

We compare the time used by all  hashing methods using \texttt{std::hash} as a reference (setting \texttt{std::hash} to 1.0).
Each single test is characterized by three ``dimensions'': \begin{inparaenum}[(1)]
\item platform and compiler; \item data; and \item  physical machine.
\end{inparaenum}
 We have reduced our analysis to the first two ``dimensions'' as follows: given a platform and data, tests were done on some number of physical machines, and  respective normalized timings were averaged.

 %We have checked that the relative standard deviation was small so that averages are meaningful.
% Further, we have calculated a relative standard deviation. In general, if relative standard deviation is within some reasonably small bounds, test results may be viewed as consistent, and averaging such results may be reasonable.

\subsubsection{Platforms, compilers, hardware, sample data}

Information about platforms/compilers used for our tests is gathered in  Appendix~\ref{appendix:platformsused}.

Performance tests were done for short (1--31~bytes) and long (256~kB) data segments.
 Results for both short and long data segments were averaged;
 all results were finally represented as ratios to the
default \texttt{std::hash} function. For these tests, data segments were provided by the SMHasher testing framework.
 None of the methods is designed or optimized for a specific type of data, such as, for instance, text, and, therefore, none of the methods is put in explicit (dis)advantage by such data generation.

We expect our results to be independent from the number of CPU cores since none of our techniques are parallelized. Moreover, we also expect the RAM type to be insignificant since even our large segments fit in L3 processor cache.

\subsubsection{Multiplatform performance results}

\begin{table*}
\centering
\caption{Relative time taken to hash (1--31~bytes) and long (256~kB)  segments (standard library $=1$). Best results are in bold.
}
\label{tab:speed}
\centering
\begin{threeparttable}[b]
\begin{tabular}{p{0.65in}p{0.55in}p{0.55in}p{0.55in}p{0.55in}p{0.55in}p{0.55in}}
\toprule
& Boost & Murmur-3\tnote{a} & SipHash & VHASH & \treenamesingle{} & \treenamedbl{}\tnote{b} \\\midrule

\multicolumn{5}{l}{\textit{x86, GCC}} \\
short &  1.03 & \textbf{0.50} & 3.94 & 2.64\tnote{c}  & 0.59 &   \\
long &  1.60 & 0.31 & 2.75  & 0.27  & \textbf{0.11} &   \\[3ex]

\multicolumn{5}{l}{\textit{x86, MSVC}} \\
short &  1.23 & \textbf{0.63} & 5.70  & 4.28\tnote{c}   & 0.76 &  \\
long &  1.11 & 0.30 & 2.39  & 0.28  & \textbf{0.10} &  \\[3ex]

\multicolumn{5}{l}{\textit{x64, GCC}} \\
short &  1.05 & 0.57 & 1.27  & 1.19\tnote{c} & \textbf{0.50}   & 0.53 \\
long &  1.38 & 0.15 & 0.73  &\textbf{0.08} & 0.10  & 0.09 \\[3ex]

\multicolumn{5}{l}{\textit{x64, MSVC}} \\
short &  1.37 & 0.77 & 1.90  & 2.02\tnote{c} & 0.69  & \textbf{0.65} \\
long &  1.40 & 0.13 & 0.67  & \textbf{0.10} & \textbf{0.09}  & \textbf{0.09} \\[3ex]

\multicolumn{5}{l}{\textit{ARMv7}} \\
short &  0.89 & \textbf{0.86} & --\tnote{d}  & 1.22\tnote{c}  & 0.87 &  \\
long &  1.29 & 0.76 & --\tnote{d}  & 0.81  & \textbf{0.49} &  \\[3ex]
\bottomrule
\end{tabular}
\begin{tablenotes}
\item [a] MurmurHash~3A for 32-bit platforms and MurmurHash~3F for 64-bit platforms. As suggested by a comment in the MurmurHash3 code~\cite{murmurhashcode}, MurmurHash~3A is best on 32-bit platforms, and MurmurHash~3F is best on 64-bit platforms (and this has been confirmed in our tests).
\item [b]  \treenamedbl{} is implemented for  64-bit platforms only
\item [c] Optimistic estimate as described above; measured values were 30--40\,\% higher.
\item [d] Not tested

\end{tablenotes}
\end{threeparttable}

\end{table*}
Results of performance testing are gathered in Table~\ref{tab:speed}.
On x86 and x64 platforms time was measured in CPU clocks (\texttt{rdtsc}); and on ARM time values were collected in microseconds. While averaging over different physical machines, the greatest relative standard deviation among all entries except SipHash was 22\,\% (deviation of SipHash was up to 38\,\%), and for over 90\,\% of entries this value was less than 15\,\%. The relative standard deviations are sufficiently small to view the presented averages as representative and to provide some assurance that relative performance results of our algorithms can be expected on a variety of platforms.
The results are summarized in Fig.~\ref{fig:per_summary}.

\begin{figure*}
\centering\includegraphics[width=0.8\textwidth]{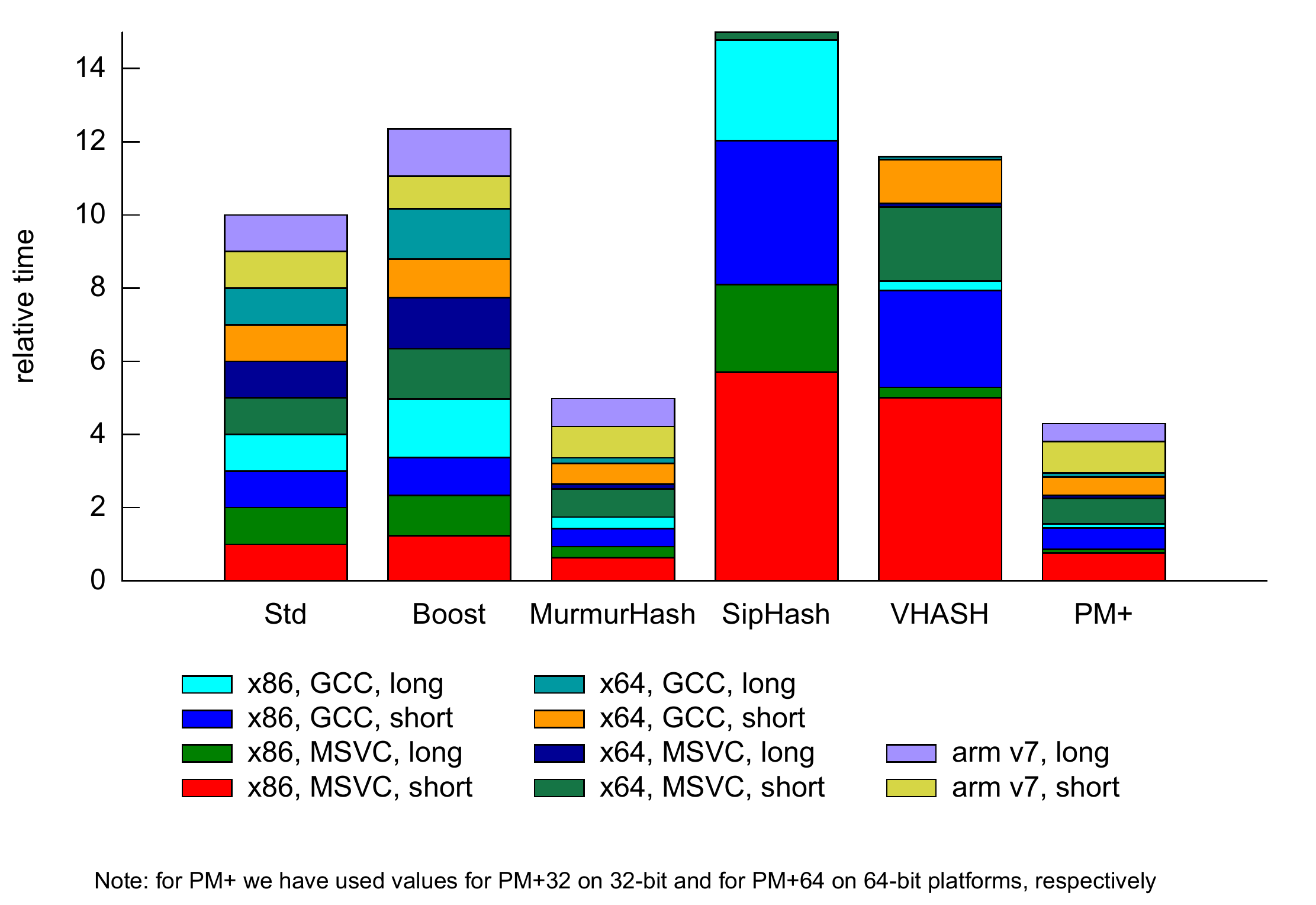}
\caption{\label{fig:per_summary}Performance  summary
}
\end{figure*}

Our results show that the hash functions in the standard library can be slow on long segments: MurmurHash,  VHASH and \treename{} can be ten times faster. But even on short segments, \treename{} can be twice as fast as the standard library (on x64 platforms).

The VHASH implementation is only competitive on long segments on x64. We are not surprised: it was designed specifically for 64-bit processors. On the x64-GCC platform, VHASH can be up to about 30\,\% faster than \treename{}. (The precise averages on long segments for VHASH and \treename{} are  0.081 and 0.106.) This is consistent with earlier
findings~\cite{Lemire10072013} (see Appendix~\ref{appendix:related}): VHASH is based on a function (NH) that is computationally inexpensive compared with \textsc{Multilinear}---at the expense of regularity.

\treename{} fares well on the ARM platform:  \treename{} is at least 50\,\% faster than the alternatives on long segments.

\treename{} is faster than MurmurHash~3 on  x64 platforms.
MurmurHash~3 is only significantly faster (20\,\%) than \treename{} on short segments on the x86-MSVC platform.

 \treenamesingle{} and  \treenamedbl{} have, in average, similar performance on 64-bit platforms. A closer examination reveals that  \treenamesingle{} is faster than \treenamedbl{} on recent processors supporting AVX2 instruction set (as reported in \S~\ref{sec:modern}) while it is slower on older processors without support for  AVX2.

\section{Conclusion}

We have described methods for constructing almost-universal  hash function
families for data strings of variable length.
Our hash functions are suitable for use
in common data structures such as hash tables. They offer strong  theoretical guarantees against denial-of-service attacks:
\begin{itemize}
\item We have almost universality: given two distinct data objects chosen by an adversary, the probability that they have the same hash value, that is, the probability that they collide, is very low given that we pick the hash functions at random. Our families have lower collision bounds than the state-of-the-art VHASH.
\item  We have shown that these hash functions are regular and component-wise regular, that is, they make an even use of all possible hash values. In doing so, they minimize the collision probability between two data objects selected at random.
%%%%
% Daniel: out of context, it is hard to tell what a sample or a component is.
%%%%%%
% Moreover, we have shown that these hash functions are %length-wise regular and even
% component-wise regular thus minimizing a collision probability of two samples that differ in only one component.
  Competitive alternatives such as VHASH are not regular which is a possible security risk~\cite{Handschuh2008,Saarinen2012}.
\end{itemize}

Further, we have shown that an implementation of these non-cryptographic hash functions offered competitive speed (as fast as MurmurHash), and  were substantially faster than the implementations found in C++ standard libraries.
 Our  approach is similar to previous work on fast universal hash families (e.g., MMH~\cite{MMH1997},  \textsc{CLHASH}~\cite{clhashing}, UMAC~\cite{703969}, VHASH~\cite{krovetz2007message}
  and Poly1305~\cite{Bernstein2005}), except that we get good regularity in addition to the high speed and universality.
To promote the use of our hash functions among practitioners and researchers, our implementation is freely available as open source software.

In the future, it may be interesting to analyze the regularity of other universal hash families~\cite{Bernstein2005,MMH1997,krovetz2001fast,clhashing}, possibly improving it when possible.
We could also seek faster families of hash functions that  are both almost universal and regular.

%
%\section{Funding}
%
%This work was supported by the Natural Sciences and Engineering Research Council of Canada [26143].

\section{Acknowledgements}

We thank Ivan Kravets for his help with our testing framework, and more specifically with the ARM processors.
\bibliographystyle{wileyj}

\bibliography{hashing}

\appendix

\section{Platforms used}
\label{appendix:platformsused}

The platforms/compilers that we have used for testing are: x86/x64 with Microsoft Visual C++ 2013 compiler; x86/x64  with GCC compiler
 (version: 4.8); and ARMv7 with the Android NDK (revision 9c, December 2013) which uses the GCC compiler internally. For more details see Table~\ref{tbl:platforms}.

\begin{table*}\caption{\label{tbl:platforms}Platforms used.
}\centering
\begin{threeparttable}[b]
\begin{tabular}{p{0.8in}p{1.4in}p{0.4in}p{0.9in}p{1.2in}}
\toprule
Name & Processor & Bits & Compiler & Flags/Configuration\\ \midrule
x64, GCC & AMD, Intel\tnote{a}      & 64 &GNU GCC 4.8       & -O2 -march=x86-64\\
x64, MSVC & Intel Core i7\tnote{b}      & 64 &MSVS 2013          & Release\\

x86, GCC & AMD, Intel\tnote{a}      & 32 &GNU GCC 4.8      & -O2 -march=i686\\
x86, MSVC & Intel Core i7\tnote{b}      & 32 &MSVS 2013          & Release\\

ARMv7 & ARM Cortex/Krait\tnote{c}    & 32 &GNU GCC 4.6\tnote{d}   & Release\\
\bottomrule
\end{tabular}
\begin{tablenotes}
\item [a] Results have been averaged over
AMD FX-8150 Eight-Core (Bulldozer, Desktop),
Intel Core i7 620M (Westmere, Mobile),
Intel Xeon E5-2630 (Sandy Bridge, Server),
Intel Core i5-3230M (Ivy Bridge, Mobile), and
Intel Core i7-4770 (Haswell, Desktop)
with a maximum relative standard deviation of 0.20.
\item [b] Results have been averaged over
Intel Core i7-2820QM (Sandy Bridge, Mobile),
Intel Core i7-3667U (Ivy Bridge, Ultra-low power),
Intel Core i7-3770 (Ivy Bridge, Desktop),
Intel Core i7-4960X (Ivy Bridge, Extreme edition), and
Intel Core i7-4700MQ (Haswell, Mobile)
with a maximum relative standard deviation of 0.12.
\item [c] Results have been averaged over
Exynos 3110 (Cortex A8),
Qualcomm Snapdragon MSM8255 (Scorpion),
dual-core Exynos 4210 (Cortex-A9),
dual-core Exynos 4412 (Cortex-A9), and
quad-core Qualcomm Snapdragon 600 (Krait 300)
with a maximum relative standard deviation of 0.21.
\item [d] From the Android NDK, revision r9d.
\end{tablenotes}
\end{threeparttable}
\end{table*}

\section{Optimization techniques for calculating scalar products on x86 and x64 processors}
\label{appendix:opttech}

As  mentioned in \S~\ref{sec:fastmul}, it is important to optimize the computation of the  scalar product.
Overall, for the computation of the scalar product on x86 processors, we found best to use vectorization in the 32-bit case presented in  \S~\ref{sec:sseopt}. In the 64-bit case, we present a thoroughly optimized use of conventional instructions  in~\S~\ref{sec:intra}.

\subsection{Vectorizing the Computation of the Scalar Product}
\label{sec:sseopt}

We can implement a scalar product over pairs of 32-bit integers using one multiplication per pair, as well as additions with carry bit (e.g., the \texttt{adc}  x86 instruction) to generate the resulting 3-word (96-bit) result.

To achieve better speed, we use the fact that modern CPUs support vector computations through Single Instruction on Multiple Data (SIMD) instructions.  For instance,  the x86 architecture has  Streaming SIMD Extensions (SSE) using 128-bit registers and the more recent Advanced Vector Extensions (AVX) using wider 256-bit registers.

Our fastest 32-bit scalar production implementation for recent Intel processors uses AVX2. AVX2 has a \texttt{vpmuludq} instruction (corresponding to the \texttt{\_mm256\_mul\_epu32} Intel intrinsic) that can multiply four pairs of 32-bit integers, thus generating four 64-bit integers.

\subsection{Faster sums using two sets of accumulators}
\label{sec:intra}

The standard instruction set may provide better  for 64-bit outputs. That is, we can multiply two 64-bit integers, and then add the 128-bit result to three 64-bit words (representing a 192-bit sum) using a sequence of x64 instructions: \texttt{mulq} (multiplication), \texttt{addq} (addition), and two \texttt{adcq} (add with a carry bit).
The steps can be described as follows:
\begin{enumerate}
\item We use three 64-bit registers as accumulators $c_1, c_2, c_3$ representing the total sum as a $3\times 64=192$-bit integer. The registers are initialized with zeros.
\item For each input pair of 64-bit values, the \texttt{mulq} instruction multiplies them and stores the results in two 64-bit registers ($a$ and $d$). One register ($a$) contains the least significant 64~bits of the product, and the other ($d$) the most significant 64~bits.
\item The first accumulator, corresponding to the least significant 64~bits, is easily updated with a simple addition (\texttt{addq}): $c_1 = (c_1 + a) \bmod\, 2^{64}$. If the sum exceeds $2^{64}-1$, the carry bit $b$ is set to 1. That is, we have that $b=  (c_1 + a) \div 2^{64}$. Then we update the second accumulator using the add-with-carry instruction (\texttt{adcq}): $c_2= (c_2 +d+b) \bmod \, 2^{64}$. We also update the third accumulator similarly.
\end{enumerate}
This approach is efficient:  we only use 4~arithmetic x64 instructions per input pair. Yet, maybe surprisingly,   there is still room for optimization.

If their operands and output values are independent, modern processors may perform more than a single instruction at a time. Indeed, recent Intel processors can retire 4~instructions (or fused $\mu{}$ops) per cycle.
 We reviewed  the initial version of our code  with the IACA code analyzer~\cite{intelIACA} for the most recent Intel microarchitecture (Haswell). IACA revealed that  the throughput was  limited by data dependencies. Though the processor can execute one multiplication per cycle, it may sometimes have to wait for the accumulators to be updated.
 Thus we rewrote our code to use two sets of accumulators. Effectively, we sum the odd terms and the even terms separately ($\sum_{i=1}^{m/2} a_{2i} s_{2i}$ and $\sum_{i=1}^{m/2} a_{2i+1} s_{2i+1}$) and then we combine them. Respective code samples can be found in   Appendix~\ref{appendix:scalar64asm} (in x64 assembly) and Appendix~\ref{appendix:scalar64cpp} (in C++ with Intel intrinsics). A new analysis with IACA reveals that the throughput of this new code is then limited by the frontend of the processor (responsible for instruction decoding). On long strings, using a recent Haswell processor (Intel i7-4770 running at 3.4\,GHz), we went from $\approx 1150$~million input pairs per second to $\approx 1350$~million input pairs per second (an 18\,\% gain).

\section{Code sample to sum 64-bit products of 32-bit integers}
\label{appendix:scalar32}

The following C++ code computes the 96-bit integer representing the sum of 128~products between pairs of 32-bit integers using AVX2 intrinsics ($m=128, n=32$). See \S~\ref{sec:sseopt} for an analysis.
For the description of the intrinsics, we refer the reader to Intel's documentation~\cite{intelintrin}.

\lstset{escapechar=@,style=customc}
\label{fig:codesample32}
\begin{lstlisting}
// input: two arrays of 32-bit integers
// const uint32_t* coeff;
// const uint32_t* x;

// output parameters:
uint64_t low_bits;
uint32_t high_bits;

__m256i ctr0, ctr1;
__m256i a, data, product, temp;
uint64_t temp_fin;

// Set accumulators to zero
ctr0 = _mm256_setzero_si256 ();
ctr1 = _mm256_setzero_si256 ();

// process the loop (unrolling may help)
for ( int i=0; i<128; i+=8 )
{
  // Load 256-bit value (eight ints)
  a = _mm256_loadu_si256
      ((__m256i *)(coeff+i));
  data = _mm256_loadu_si256
      ((__m256i *)(x+i));
  // multiply ints at even positions
  product = _mm256_mul_epu32 ( data, a);
  temp = _mm256_srli_epi64
      ( product, 32 );
  ctr1 = _mm256_add_epi64
      ( ctr1, temp );
  ctr0 = _mm256_add_epi64
      ( ctr0, product );
  // exchange even-odd
  // note: 0xb1 =  1*1+0*4+3*16+2*64
  a = _mm256_shuffle_epi32
      ( a, 0xb1);
  data = _mm256_shuffle_epi32
      ( data, 0xb1 );
  // multiply ints at even positions
  // (former odd positions)
  product = _mm256_mul_epu32 ( data, a);
  temp = _mm256_srli_epi64
      ( product, 32 );
  ctr1 = _mm256_add_epi64
      ( ctr1, temp );
  ctr0 = _mm256_add_epi64
      ( ctr0, product );
}

// finalize

// desired results are in c0 and c1
// we interleave the sums and add them
temp = _mm256_unpackhi_epi64
      ( ctr0, ctr1 );
data = _mm256_unpacklo_epi64
      ( ctr0, ctr1 );
ctr1 = _mm256_add_epi64
      ( data, temp );
// extract  a 64+32 bit number
// (low_bits, high_bits)
uint64_t lo = *(uint64_t*)(&ctr1) +
            ((uint64_t*)(&ctr1))[2];
uint64_t hi = ((uint64_t*)(&ctr1))[1] +
            ((uint64_t*)(&ctr1))[3];
uint32_t lohi = lo >> 32;
uint32_t hilo = hi;
uint32_t diff = lohi - hilo;
hi += diff;
lo = (uint32_t)lo +
      (((uint64_t)(uint32_t)hi)<<32);

// answer:
low_bits = lo;
high_bits = hi >> 32;

\end{lstlisting}

\section{Code sample to sum 128-bit products of 64-bit integers (assembler)}
\label{appendix:scalar64asm}

The following assembly code computes the 192-bit integer representing the sum of 128~products between pairs of 64-bit integers using the standard x64 instruction set ($m=128, n=64$).
\lstset{escapechar=@,style=customc}

\begin{lstlisting}
// input: pointers to 64-bit arrays
// rbx: address of the start of
//      the 1st array
// rcx: address of the start of
//      the 2nd array
// 1st accumulator:
// r10: least significant 64 bits
// r11: mid 64 bits
// r12: most significant 64 bits
// 2nd accumulator:
// r13: least significant 64 bits
// r14: mid 64 bits
// r15: most significant 64 bits

// add 1st product to 1st accumulator
movq 0(%rbx),%%rax\n
mulq 0(%rcx)\n
addq %%rax,  %%r10\n
adcq %%rdx,  %%r11\n
adcq $0,  %%r12\n

// add 1st product to 1st accumulator
movq 8(%rbx),%%rax\n
mulq 8(%rcx)\n
addq %%rax,  %%r13\n
adcq %%rdx,  %%r14\n
adcq $0,  %%r15\n

// ... repeat as necessary

// merge accumulators:
movq 8(%rbx),%%rax\n
mulq 8(%rcx)\n
addq %%rax,  %%r13\n
adcq %%rdx,  %%r14\n
adcq $0,  %%r15\n

// the sum of products is now at
// (r10, r11, r12)

\end{lstlisting}

\section{Code sample to sum 128-bit products of 64-bit integers using Intel intrinsics}
\label{appendix:scalar64cpp}

The following C++ code computes the 192-bit integer representing the sum of 128~products between pairs of 64-bit integers using Intel intrinsics ($m=128, n=64$). Such code is well suited for the Microsoft Visual C++ compiler.
\lstset{escapechar=@,style=customc}

\begin{lstlisting}
// 1st accumulator:
uint64_t low1   = 0; // least sign. 64 bits
uint64_t high1  = 0; // next 64 bits
uint64_t vhigh1 = 0; // most significant
// 2nd accumulator
uint64_t low2   = 0; // least sign. 64 bits
uint64_t high2  = 0; // next 64 bits
uint64_t vhigh2 = 0; // most significant
// intermediates:
uint64_t mulLo, mulHi;
unsigned char c;
for(size_t i = 0; i<128; i+=2)
{
  // process even pair
  // _umul128 is Microsoft-specific
  mulLo = _umul128(a[i],s[i],&mulHi);
  // _addcarry_u64 is an Intel intrinsic
  // supported by Microsoft
  c = _addcarry_u64
     (0, mulLo, low1, &low1);
  c = _addcarry_u64
     (c, mulHi, high1, &high1);
  _addcarry_u64(c, vhigh1, 0, &vhigh1);

  // process odd pair
  mulLo = _umul128
         (a[i+1],s[i+1],&mulHi);
  c = _addcarry_u64
         (0, mulLo, low2, &low2);
  c = _addcarry_u64
         (c, mulHi, high2, &high2);
  _addcarry_u64(c, vhigh2, 0, &vhigh2);
}

c = _addcarry_u64(0, low1, low2, &low1);
c = _addcarry_u64
    (c, high1, high2, &high1);
_addcarry_u64
    (c, vhigh1, vhigh2, &vhigh1);

// result is at (low1, high1, vhigh1)


\end{lstlisting}

\section{Non-Regularity of the VHASH family}
\label{appendix:related}

The effort to design practical universal random hash functions with good properties has a long history. Thorup~\cite{338597} showed that strongly universal hashing could be very fast.
Crosby and Wallach~\cite{Crosby:2003:DSV:1251353.1251356} showed that almost universal hashing could be as fast as common deterministic hash functions. Their conclusion was that while universal hash functions were not standard practice, they  should be.
In particular, they got good experimental results with UMAC~\cite{703969}.

%\subsection{VHASH}
%\label{sec:vhash}

More recently, Krovetz proposed the VHASH family~\cite{krovetz2007message}. On 64-bit processors, it is
faster than the hash functions from  UMAC.

Like UMAC, VHASH is $\epsilon$-almost $\Delta$-universal and builds on the NH family:
\begin{align*}
\mathrm{NH}(s)=\Bigg ( \sum_{i=1}^{l/2}\big (&((s_{2i-1}+k_{2i-1}) \bmod\, 2^n) \\
&\times ((s_{2i}+k_{2i}) \bmod \, 2^n)\big ) \Bigg ) \bmod \, 2^{2n}.
\end{align*}
NH is fast in part due to the fact that it uses one multiplication per pair of input words. In contrast, MMH or
\treename{}, as  derivatives of \textsc{Multilinear}, use at least one multiplication per input word. However, the number of multiplications is not necessarily a performance bottleneck: recent Intel processors can execute one multiplication per cycle. We should not expect hash functions with half the number of multiplications to be twice as fast~\cite{Lemire10072013}. For example, a fast hash function might be limited by the number of micro-operations that the processor can retire per cycle (4 on recent Intel processors) rather than by the number of multiplications.

Like \textsc{Multilinear}, NH is $1/2^{n}$-almost $\Delta$-universal, but \textsc{Multilinear} generates values in $[0,2^n)$ whereas NH generates values in $[0,2^{2n})$. (For this reason, NH might not be well suited for a tree-based approach as described in \S~\ref{sec:pyr}.)

Krovetz  reports that VHASH is twice as fast as UMAC (0.5~CPU cycle per input byte vs.\ 1~CPU cycle per input byte on an AMD Athlon processor).
For long strings on 64-bit processors, we expect VHASH to be one of the fastest universal hash families.

The updated VHASH~\cite{dai2007vhash} family is $\epsilon$-almost universal over  $[0,2^{64}-257)$ with $\epsilon= \frac{1}{2^{61}}$ for strings of length up to $2^{62}$~bits.
 In contrast, \treename{} produces hash values in  $[0,2^{64})$ with $\epsilon =\frac{8}{2^{63}-6}$ for strings of length up to   $(2^{62}-64)$~bits.

We can describe the 64-bit VHASH as follows:  NH is used with $n=64$ to generate 128-bit hash values on 128-byte blocks. The result is $1/2^{64}$-almost $\Delta$-universal on each block. In turn, the result is mappend to the interval $[0,2^{126})$ by applying a  modulo reduction ($\bmod\,2^{126}$): the family is then $1/2^{62}$-almost $\Delta$-universal on each block. The hashed values over each block are then aggregated using a polynomial hash family computed over $[0,2^{127}-1)$. The result is finally reduced to $[0,2^{64}-257)$ with modulo operations and divisions.

The NH family is not regular.
For instance, consider values of $s$ where $s=(s_1, s_2)$ and at least one of $s_1$ and $s_2$ is even, which is $\frac{3}{4}$ of all possible values. If both $k_1$ and $k_2$ are even, then $\mathrm{NH}(s)$ is even, too, and, therefore, $\frac{3}{4}$ of all values are mapped to only $\frac{1}{2}$ of all values.

To make matters worse, the NH family is never regular for any choice of keys ($k_i$) and it has ``very little regularity'' as $n$ grows in the following sense. For any given integers $k_1, k_2 \in [0,2^n)$, consider the map from $[0,2^n) \times [0,2^n) \to [0,2^{2n})$ given by
$\mathrm{NH}(x,y)=((x+k_{1} \bmod \, 2^n) ((y+k_{2} \bmod \, 2^n) \bmod \, 2^{2n}$.
Because we pick $x,y \in [0,2^n)$, we can choose $k_1=k_2=0$ without loss of generality. We can then ask about the size of the image of $\mathrm{NH}(x,y)$. That is, which fraction of all integers in $[0,2^{2n})$ are the product of two numbers in $[0,2^n)$? Erd\"os showed that this ratio goes to zero as $n$ becomes large~\cite{erdos}. Though we do not know of an exact formula, we plot the relative size of the image of NH in Fig.~\ref{fig:nhsize}: already at $n=20$ only about one integer out of five in $[0,2^{40})$ can be generated by the product of two integers in $[0,2^{20})$~\cite{A027417}. We expect that for $n=64$, the ratio is considerably less than 20\,\%. Note that keeping only, say, the least significant  $2n-2$~bits (e.g., applying $\mod  2^{2n-2}$) or most significant $2n-2$ bits (e.g., applying $\div 2^2$) does not change the core result: the relative size of the image still goes to zero as $n$ becomes large.

Hence NH is not even 5-regular.
The issue is more dramatic if we consider component-wise regularity. Indeed, consider $\mathrm{NH}(s)$ over 2-character strings $(s_1,s_2)$. If  $s_2+k_2 \bmod{2^{n}}=0$,
we have that $\mathrm{NH}(s)= 0$ for all values of $s_1$, which is the worst possible case. NH and VHASH are not at all component-wise regular.

\begin{figure}
\centering\includegraphics[width=0.7\columnwidth]{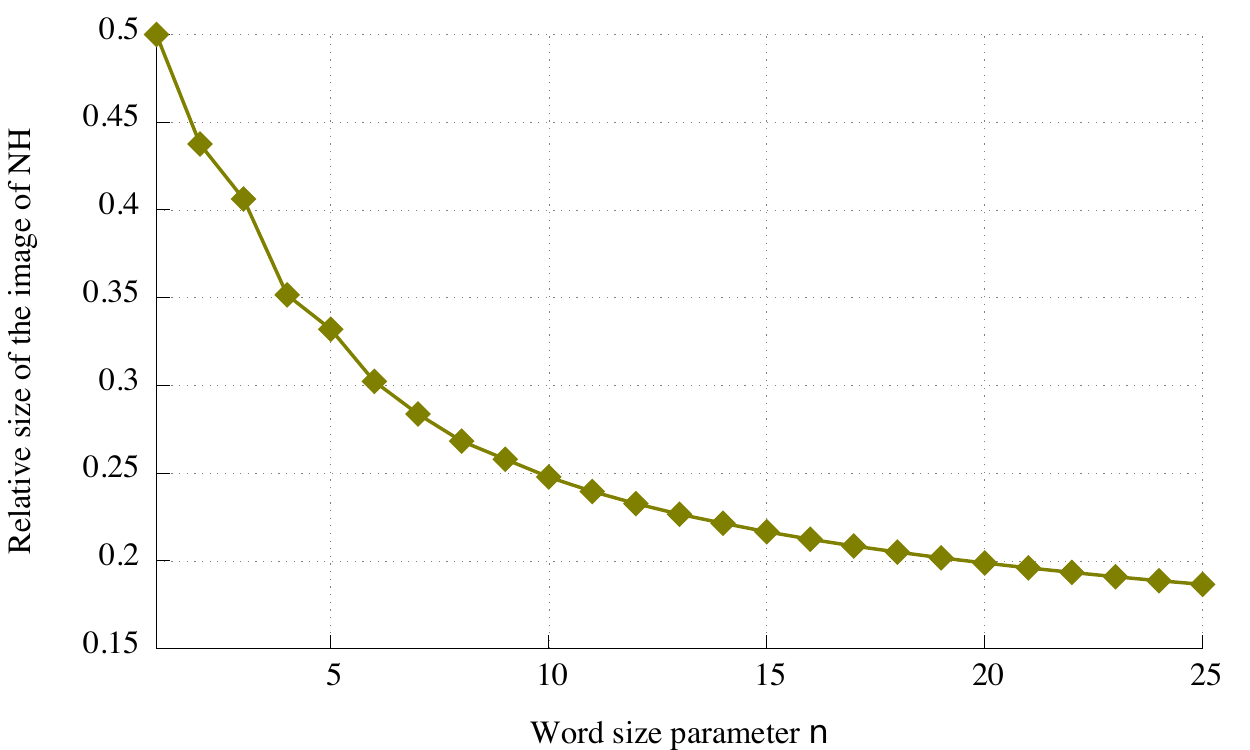}
\caption{\label{fig:nhsize}Fraction of all $2n$-bit integers that are the product of two $n$-bit integers. }
\end{figure}

\end{document}